\documentclass[a4paper,10pt,abstracton]{scrartcl}
\usepackage{amsthm}
\usepackage{amsfonts}
\usepackage{amsmath}
\usepackage{amssymb}
\usepackage{fixltx2e}
\usepackage{hyperref}
\usepackage{mathrsfs}
\usepackage{marvosym}
\usepackage{nicefrac}
\usepackage{array}
\usepackage{color}
\usepackage{authblk}

\newlength{\Lcolw}
\newlength{\Lfcolw}
\newcolumntype{L}{%
  >{\raggedleft\hspace{0pt}}p{\Lcolw}}%
\newsavebox{\Lone}
\sbox{\Lone}{\parbox{\Lfcolw}{\centering 1}}

\renewcommand{\mathcal}[1]{\mathscr{#1}}

\setkomafont{title}{\normalfont}
\setkomafont{disposition}{\normalfont\bfseries}

\theoremstyle{plain}
\newtheorem{theorem}{Theorem}
\newtheorem{lemma}[theorem]{Lemma}
\newtheorem{proposition}[theorem]{Proposition}
\newtheorem{corollary}[theorem]{Corollary}
\newtheorem{claim}[theorem]{Claim}

\theoremstyle{definition}
\newtheorem{definition}[theorem]{Definition}
\newtheorem*{example}{Example}

\theoremstyle{remark}
\newtheorem*{remark}{Remark}

\DeclareMathOperator{\val}{val}
\DeclareMathOperator{\adj}{adj}
\DeclareMathOperator{\num}{num}
\DeclareMathOperator{\sgn}{sgn}

\DeclareMathOperator{\rev}{rev}

\newcommand{\transpose}{\ensuremath{\mathsf{T}}}
\newcommand{\allones}{\ensuremath{\mathbf{1}}}
\newcommand{\toep}{\ensuremath{\mathbf{T}}}
\newcommand{\abs}[1]{\ensuremath{\mathopen\lvert #1 \mathclose\rvert}}

\newcommand{\symdiff}{\mathbin{\triangle}}
\newcommand{\patwl}{\ensuremath{\tau_{\textrm{wl}}}}
\newcommand{\patwld}{\ensuremath{\tau_{\textrm{wld}}}}
\newcommand{\RR}{\mathbb{R}}

\newcommand{\NP}{\ensuremath{\mathsf{NP}}}
\newcommand{\PLS}{\ensuremath{\mathsf{PLS}}}

\newcommand{\calI}{\ensuremath{\mathcal{I}}}
\newcommand{\calF}{\ensuremath{\mathcal{F}}}
\newcommand{\calN}{\ensuremath{\mathcal{N}}}

\newcommand{\MAXCUT}{\ensuremath{\textsc{MaxCut}}}
\newcommand{\FLIP}{\ensuremath{\textsc{Flip}}}
\newcommand{\TSP}{\ensuremath{\textsc{Travelling Salesman}}}
\newcommand{\OPT}{\ensuremath{\textsc{Opt}}}

\begin{document}

\title{Patience of Matrix Games}


\author[1]{\mbox{Kristoffer Arnsfelt Hansen}}
\author[1]{\mbox{Rasmus Ibsen-Jensen}}
\author[2]{\mbox{Vladimir V. Podolskii}}
\author[3]{\mbox{Elias Tsigaridas}}

\affil[1]{Aarhus University\thanks{Hansen and Ibsen-Jensen acknowledge
    support from the Danish National Research Foundation and The
    National Science Foundation of China (under the grant 61061130540)
    for the Sino-Danish Center for the Theory of Interactive
    Computation, within which this work was performed. They also
    acknowledge support from the Center for Research in Foundations of
    Electronic Markets (CFEM), supported by the Danish Strategic
    Research Council.}}
\affil[2]{Steklov Mathematical Institute\thanks{Part of the research was done during a visit to Aarhus University.
The research is partially supported by the Russian Foundation for Basic Research
and the programme ``Leading Scientific Schools''.}}
\affil[3]{INRIA Paris-Rocquencourt and Universite Pierre et Marie Curie (Paris 6)%
\thanks{Partially supported from the EXACTA grant of the National Science
  Foundation of China (NSFC 60911130369) and the French National
  Research Agency (ANR-09-BLAN-0371-01).}
}

\date{\normalsize\today}

\maketitle
\begin{abstract}
  For matrix games we study how small nonzero probability must be used
  in optimal strategies. We show that for $n \times n$ win-lose-draw
  games (i.e.\ $(-1,0,1)$ matrix games) nonzero probabilities smaller
  than $n^{-O(n)}$ are never needed. We also construct an explicit $n
  \times n$ win-lose game such that the unique optimal strategy uses a
  nonzero probability as small as $n^{-\Omega(n)}$. This is done by
  constructing an explicit $(-1,1)$ nonsingular $n \times n$ matrix,
  for which the inverse has only nonnegative entries and where some of
  the entries are of value $n^{\Omega(n)}$.

\end{abstract}

\section{Introduction}
Given a matrix game $A$ we are interested in the following question:
What is the smallest nonzero probability that must be used in optimal
strategies. This quantity, the smallest nonzero probability of a
strategy, was first considered in the context of recursive games
(stochastic games where payoffs are only accumulated in absorbing
states) by Everett~\cite{AMS:Everett57}. To be more precise, if $p$
is the smallest nonzero probability of a probability vector $\sigma$,
we say that the \emph{patience} of $\sigma$ is $1/p$. Note that this
is the precisely the expected number of times that $\sigma$ must be
sampled in order to observe the least likely outcome. Also
$\lceil\log_2(1/p)\rceil$ is a lower bound on the number of random
bits required in order to sample from $\sigma$ using a source of
uniform random bits.

In this paper we study the patience required for playing optimal
strategies in matrix games. Our focus is how this quantity depends on
the dimensions of the matrix game, rather than on the individual
payoffs. In particular we consider win-lose and win-lose-draw matrix
games. We model win-lose games as $(0,1)$ matrices and win-lose-draw
matrices as $(-1,0,1)$ matrices. Note that for win-lose games this
choice of matrices have no consequence: the set of optimal strategies
is invariant under addition by a number and multiplication by a
positive number, applied simultaneously to every entry of the
matrix. In particular, we can equivalently model win-lose matrix games
as $(-1,1)$ matrices.

We prove both upper and lower bounds on the patience required for
playing optimal strategies for these two classes of matrix games. Our
lower bounds build on previous constructions of \emph{ill-conditioned}
matrices~\cite{LAA:GS84,JCT:AV97}. In particular we show that from any
ill-conditioned matrix $A$, a matrix game can be derived with patience
at least the size of the largest entry of the inverse of $A$. As such
our question can be seen as yet another application of ill-conditioned
matrices. A downside of this connection is that it is not
\emph{explicit} - namely, we do not know of a polynomial time
algorithm for computing this derived matrix game, given the
ill-conditioned matrix $A$ as input. We address this unsatisfactory
situation by constructing a variant of the ill-conditioned matrix
constructed by Alon and V\~{u}~\cite{JCT:AV97}, and study in detail
the structure of the inverse matrix. We use this to construct an
ill-conditioned $(-1,1)$ matrix with a \emph{non-negative inverse},
and from this we directly obtain an explicit construction of a
win-lose matrix of high patience. This construction is in fact what we
will call \emph{fully-explicit}, meaning that each entry of the matrix
can be computed in time polynomial in the bitlength of the dimension
of the matrix.

Patience behaves very differently in matrix games compared to its
original setting of recursive
games~\cite{LICS:HKM09,CSR:HIM11,STOC:HKLMT11}. First of all in
recursive games optimal strategies are not guaranteed to exist; On the
other hand, for every $\epsilon>0$ both players have stationary
strategies guaranteeing an expected payoff within $\epsilon$ of the
value of the game from any starting
position~\cite{AMS:Everett57}. However, there exist recursive games
with $N$ positions, each with $m \geq 2$ actions, and every payoff is
0 or 1, where any $(1-4m^{-N/2})$-optimal strategy must have patience
at least $2^{m^{N/3}}$. On the other hand, patience
$(1/\epsilon)^{m^{O(N)}}$ is always sufficient for having an
$\epsilon$-optimal strategy in recursive games where every payoff is 0
or 1.

Consider now again the setting of matrix games.  Lipton and Young
proved that in a zero-sum $n\times n$ matrix game where all payoffs
belongs to the interval $[0,1]$, each player has a \emph{simple}
strategy guaranteeing an expected payoff within $\epsilon$ of the
value of the game, where a simple strategy is a strategy that mixes
uniformly on a multiset of $\lceil\ln n/\epsilon^2\rceil$
actions~\cite{STOC:LY94}. Thus the patience of such strategies is also
at most $\lceil\ln n/\epsilon^2\rceil$.

In other words, comparing with our results below, if one is willing to
give up $\epsilon$ payoff, one can play with patience that is smaller
by an exponential magnitude than required for playing truly optimally.

\subsection{Our Results}

For stating our results we use standard matrix game terminology. We
refer the reader to Section~\ref{SEC:Preliminaries} for explanations
of this terminology.

We define the \emph{patience} of a strategy to be $1/p$, where $p$ is
the smallest nonzero probability the strategy $x$ assigns to one of
the actions. That is, the patience of a strategy $x$ is $(\min\{x_i
\mid x_i >0\})^{-1}$.

Given a matrix game $A$, we define the \emph{patience} $\tau^1(A)$
required for Player 1 to play $A$ optimally to be the minimum patience
of an optimal strategy for $A$. In a similar way we can define
$\tau^2(A)$ for Player 2.

We are interested in the largest patience required for optimal
strategies of win-lose and win-lose-draw games, as a function of the
size of the matrix game.  Thus, define $\patwl(n)$ as the maximum of
$\tau^1(A)$ taken over all $(0,1)$ $n \times n$ matrix games $A$, and
similarly $\patwld(n)$ as the maximum of $\tau^1(A)$ taken over all
$(-1,0,1)$ $n \times n$ matrix games $A$.

Clearly the definition of $\patwl(n)$ and $\patwld(n)$ would be
unchanged by considering $\tau^2(A)$ rather than $\tau^1(A)$. However
we shall also consider the patience required by both players for
optimal strategies. Thus, we define also $\widehat{\patwl}(n)$ as the
maximum of $\min(\tau^1(A),\tau^2(A))$ taken over all $(0,1)$ $n
\times n$ matrix games $A$, and similarly $\widehat{\patwld}(n)$ as
the maximum of $\min(\tau^1(A),\tau^2(A))$ taken over all $(-1,0,1)$
$n \times n$ matrix games $A$. All these measures of patience are in
fact closely related (cf.\ Section~\ref{SEC:Patience-preliminaries}).

We are now able to state our results. First using a Theorem of Shapley
and Snow~\cite{AMS:SS50} and standard estimates of the magnitude
determinants we obtain the following basic upper bound on patience.
\begin{proposition}
\label{PROP:PatienceUpperbound}
\[
\patwl(n) \leq (n+2)^\frac{n+2}{2}/2^{n+1} \enspace ,  \quad
\patwld(n) \leq (n+1)^\frac{n+1}{2}  \enspace .
\]
\end{proposition}

Next, using previous results on ill-conditioned matrices by Alon and
V\~{u} \cite{JCT:AV97} we obtain the following (non-explicit) lower
bound on patience.
\begin{theorem}
\label{THM:IndirectPatienceLowerbound}
\[
\widehat{\patwld}(n) \geq n^{\frac{n}{2}}/2^{n(2+o(1))} \enspace .
\]
\end{theorem}
By Corollary~\ref{COR:patience-relationship-with-draw} from
Section~\ref{SEC:Patience-preliminaries} we obtain a result for
win-lose matrix games as well.
\begin{corollary}
\label{COR:IndirectPatienceLowerbound}
\[
\widehat{\patwl}(n) \geq n^{\frac{n}{4}}/2^{n(5/4+o(1))} \enspace .
\]
\end{corollary}

Our main contribution is an explicit construction of a matrix game
satisfying a similar patience lower bound.
\begin{theorem}
\label{THM:MainLowerBound}
  Let $n=2^m$ be a power of two. Then
\[
\patwl(n) \geq n^{\frac{n}{2}}/2^{n(1+o(1))} \enspace .
\]
Furthermore there is an algorithm that for each $n$ and given indices
$i$ and $j$ computes the entry $(i,j)$ of the matrix witnessing the
lower bound, in time polynomial in $m$.
\end{theorem}

\subsection{Organization of the paper}

In Section~\ref{SEC:Preliminaries} we briefly introduce sign patterns
of matrices and matrix games, followed by a more extensive coverage of
patience of matrix games in
Section~\ref{SEC:Patience-preliminaries}. In particular this section
provides a proof of the upper bound on patience of matrix games,
Proposition~\ref{PROP:PatienceUpperbound}.  In
Section~\ref{SEC:Patience-and-Ill-conditioned-Matrices} we consider
the relationship between ill-conditioned matrices and patience of
matrix games. In Section~\ref{SEC:Exlicit-Exponential-Exaxmples} we
consider three easy examples of explicit ill-conditioned matrices and
show how they give matrix games of large patience. Finally, in
Section~\ref{SEC:Almost-Worst-case} we present our main contribution,
an explicit construction of a win-lose matrix game of almost worst
case patience.

\section{Preliminaries}
\label{SEC:Preliminaries}

We shall denote by $\allones$ a vector of appropriate dimension where
every entry is 1. All vectors we consider are column vectors.

\subsection{Sign Patterns of Matrices}

A \emph{full sign pattern} is a matrix with entries from $\{-1,1\}$. A
pair of vectors $\sigma^{(1)},\sigma^{(2)}$ with entries from
$\{-1,1\}$ gives rise to a full sign pattern $\sigma^{(1)}
(\sigma^{(2)})^\transpose$ We shall call a full sign pattern of this form a
\emph{block checkerboard sign pattern}.

Let $A=(a_{ij})$ be a $n \times n$ matrix with real valued entries. We
say that $A$ weakly obeys a block checkerboard sign pattern if there
is a block checkerboard sign pattern $\Sigma=(\sigma_{ij})$ such that
$\sigma_{ij} = 1$ implies $a_{ij} \geq 0$ and $\sigma_{ij} = -1$
implies $a_{ij} \leq 0$. Note that given $A$, $\Sigma$ is not
necessarily unique, depending upon the entries of $A$ that are 0.

\begin{lemma}
\label{LEM:Checkerboard}
  Let $A=(a_{ij})$ be a $n \times n$ nonsingular matrix with real
  valued entries, such that the inverse $A^{-1}$ weakly obeys a block
  checkerboard sign pattern $\Sigma = (\sigma_{ij})$. Define the $n
  \times n$ matrix $B = A \circ \Sigma^\transpose$ to be the Hadamard
  product of $A$ and the transpose of $\Sigma$ (That is, $B=(b_{ij})$
  is given by $b_{ij} = a_{ij} \sigma_{ji}$). Then the entries of the
  inverse $B^{-1}$ are non-negative.
\end{lemma}
\begin{proof}
This follows immediately by considering the identity $AA^{-1} = I$.
\end{proof}

\subsection{Matrix Games}

A \emph{matrix game} is given by a $m \times n$ real matrix $A =
(a_{ij})$. The entries $a_{ij}$ are \emph{payoffs}. The game is played
by Player 1 selecting an action $i \in \{1,\dots,m\}$ and Player 2
simultaneously selecting an action $j \in \{1,\dots,n\}$, after which
Player 1 receives a payoff of $a_{ij}$ from Player 2.  A
\emph{strategy} of a player is a probability distribution over the
actions of the player. We shall view these as stochastic vectors. A
strategy is \emph{totally mixed} if it assign non-zero probability to
each action.

Given a strategy $x$ for Player 1 and a strategy $y$ for Player 2,
the expected payoff to Player 1 when the two players use the pair
$(x,y)$ of strategies is $x^\transpose A y$. The celebrated minimax
theorem of von~Neumann~\cite{MA:vonNeumann28} states that every matrix
game has a \emph{value}.
\begin{theorem}[von~Neumann]
For any matrix game $A$, there is a number $v$ such that
\[
v = \max_x \min_y x^\transpose A y = \min_y \max_x x^\transpose A y \enspace ,
\]
where $x$ and $y$ range over strategies for the two players. The
number $v$ is called the value, $\val(A)$, of $A$.
\end{theorem}
A strategy $x$ is a \emph{maximin} strategy for Player 1, if $\min_y
x^\transpose A y = \val(A)$. Similarly, a strategy $y$ is a
\emph{minimax} strategy for Player 2, if $\max_x x^\transpose A y =
\val(A)$. We will call both a maximin strategy for Player 1 and a minimax
strategy for Player 2 for \emph{optimal} strategies.

Shapley and Snow~\cite{AMS:SS50} characterized the set of optimal
strategies as the convex hulls of \emph{basic solutions}.
\begin{theorem}[Shapley and Snow]
\label{THM:ShapleySnow}
Let $X$ and $Y$ be the sets of optimal strategies for Player 1 and
Player 2 in a matrix game. Then $X$ and $Y$ are the convex hulls of
the sets of \emph{basic solutions} $X^*$ and $Y^*$, where every pair
of basic solutions $x \in X^*$ and $y \in Y^*$ correspond exactly to a
square submatrix $B$ of $A$, which satisfies:
\begin{equation}
\label{EQ:ShapleySnow}
\begin{split}
  \val(A) & = \det(B)/\allones^\transpose \adj(B) \allones \enspace ,\\
  x_B^\transpose & = \allones^\transpose \adj(B) / \allones^\transpose \adj(B) \allones \enspace ,  \\
  y_B & = \adj(B) \allones / \allones^\transpose \adj(B) \allones
  \enspace ,
\end{split}
\end{equation}
where $x_B$ and $y_B$ are obtained from of $x$ and $y$ by restricting
to the rows and columns of $B$, respectively.
\end{theorem}

If the value $v$ of the matrix game $A$ is nonzero this simplifies to
\begin{equation}
\label{EQ:ShapleySnowInvertible}
\val(A) = 1/\allones^\transpose B^{-1} \allones \enspace , \quad
x_B^\transpose = v \allones^\transpose B^{-1} \enspace , \quad
y_B = v B^{-1} \allones \enspace ,
\end{equation}

Conversely, we have the following result (see e.g.\ \cite[Theorem 3.2]{FergusonNotes}).
\begin{theorem}
\label{THM:NonsigularMatrixGame}
  Let $A$ be a $n \times n$ matrix game, where $A$ is nonsingular and
  $\allones^\transpose A^{-1} \allones \neq 0$. Define
\[
v = 1/\allones^\transpose A^{-1} \allones \enspace , \quad
x^\transpose = v \allones^\transpose A^{-1} \enspace , \quad
y = v A^{-1} \allones \enspace .
\]
If both $x \geq 0$ and $y \geq 0$ then $\val(A)=v$ and $x$ and $y$ are
optimal strategies of $A$. If in fact both $x$ and $y$ are totally
mixed, i.e.\ $x > 0$ and $y > 0$, then $x$ and $y$ are the unique
optimal strategies.
\end{theorem}

\subsection{Patience of matrix games}
\label{SEC:Patience-preliminaries}

\subsubsection*{Basic Relations}
Directly from the definitions we have:
\begin{proposition}
\label{PROP:patiencefor2}
\[\widehat{\patwl}(n) \leq \patwl(n) \enspace , \quad
\widehat{\patwld}(n) \leq \patwld(n)  \enspace .
\]
\end{proposition}

Conversely we can from matrix games where one player must use a
strategy of high patience construct a (larger) matrix game where this
is the case for both players.

\begin{proposition}
\label{PROP:patiencefor2-(0,1)}
  Let $A$ be a $n \times n$ $(0,1)$ matrix game such that
  $0<\val(A)<1$. Then there exist a $2n \times 2n$ $(0,1)$ matrix game
  $B$ such that $\widehat{\patwl}(B) \geq \max(\tau^1(A),\tau^2(A))$.
\end{proposition}
\begin{proof}
  Consider the $2n \times 2n$ matrix game $B$ given by
\[
B =
\begin{bmatrix}
A & 0\\
0 &  \allones\allones^\transpose-A^\transpose
\end{bmatrix} \enspace .
\]
Note that
$\val(\allones\allones^\transpose-A^\transpose)=1-\val(A)>0$. It
follows that the optimal strategies for Player 1 in $B$ are of the
form $\left((1-\val(A))x^\transpose,\val(A)y^\transpose\right)$ and
similarly the optimal strategies for Player 2 in $B$ are of the form
$\left((1-\val(A))y^\transpose,\val(A)x^\transpose\right)$, where $x$
and $y$ are optimal strategies in $A$ for Player 1 and Player 2, and
the result follows.
\end{proof}

We have a similar statement for win-lose-draw matrix games.
\begin{proposition}
\label{PROP:patiencefor2-(-1,0,1)}
  Let $A$ be a $n \times n$ $(-1,0,1)$ matrix game such that
  $-1<\val(A)<1$. Then there exist a $2n \times 2n$ $(-1,0,1)$ matrix
  game $B$ such that $\widehat{\patwl}(B) \geq
  \max(\tau^1(A),\tau^2(A))$.
\end{proposition}
\begin{proof}
  The proof follows similarly to that of
  Proposition~\ref{PROP:patiencefor2-(0,1)}, by considering the $2n
  \times 2n$ matrix game $B$ given by
\[
B =
\begin{bmatrix}
A & -\allones\allones^\transpose\\
-\allones\allones^\transpose &  -A^\transpose
\end{bmatrix} \enspace .
\]
and noting that this matrix game has the same optimal strategies as
the matrix game obtained by adding 1 to each entry:
\[
\begin{bmatrix}
A+\allones\allones^\transpose & 0\\
0 &  \allones\allones^\transpose-A^\transpose
\end{bmatrix} \enspace ,
\]
$\val(A+\allones\allones^\transpose)=\val(A)+1>0$, and $\val(\allones\allones^\transpose-A^\transpose)=1-\val(A)>0$.
\end{proof}

Since a win-lose matrix game of value $0$ or $1$ as well as a
win-lose-draw matrix game of value $-1$ or $1$ has trivial patience
$1$ for both players we have the following relations, complementing
Proposition~\ref{PROP:patiencefor2}.
\begin{corollary}
\[\widehat{\patwl}(2n) \geq \patwl(n) \enspace , \quad
\widehat{\patwld}(2n) \geq \patwld(n)  \enspace .
\]
\end{corollary}

Next we consider the relationship between win-lose and win-lose-draw
games. Immediately from the definition we have.
\begin{proposition}
\label{PROP:patience-relationship-with-draw}
\[
\patwl(n) \leq \patwld(n) \enspace , \quad
\widehat{\patwl}(n) \leq \widehat{\patwld}(n) \enspace .
\]
\end{proposition}
We next show how to convert win-lose-draw matrix games into win-lose
matrix games.
\begin{proposition}
\label{PROP:win-lose-draw-to-win-lose}
  Let $A$ be a $n \times n$ $(-1,0,1)$ matrix game. Define the $2n
  \times 2n$ $(0,1)$ matrix game $B$ obtained from $A$ by replacing
  $(-1)$-entries by the $2\times 2$ all-zero matrix, $0$-entries by
  the $2\times 2$ identity matrix, and $1$-entries by the $2\times 2$
  all-ones matrix. Then $\tau^1(A) \leq \tau^1(B)$ and $\tau^2(A) \leq
  \tau^2(B)$
\end{proposition}
\begin{proof}
  Let $x'$ be an optimal strategy for Player~1 in $B$. Define the
  strategy $x$ for Player~1 in $A$ by $x_i = x'_{2i-1}+x'_{2i}$. By definition we
  have $(x'^\transpose (2B-\allones\allones^\transpose))_j =
  2(x'^\transpose B)_j -1 \geq 2\val(B)-1$ for all $j$. We then get
\[
\begin{split}
(x^\transpose A)_j & = \frac{1}{2}(x'^\transpose (2B-\allones\allones^\transpose))_{2j-1} +
\frac{1}{2}(x'^\transpose (2B-\allones\allones^\transpose))_{2j} \\
& \geq \frac{1}{2}(2\val(B)-1) +  \frac{1}{2}(2\val(B)-1) = 2\val(B)-1
\end{split}
\]
Similarly, for an optimal strategy $y'$ for Player 2 in $B$ we define
the strategy $y$ for Player 2 in $A$ by $y_i = y'_{2i-1}+y'_{2i}$ and
obtain $(Ay)_i \leq 2\val(B)-1$ for all $i$. It follows that
$\val(A)=2\val(B)-1$ and $x$ and $y$ are optimal strategies in
$A$. Since the patience of $x$ is at most the patience of $x'$ and the
patience of $y$ is at most the patience of $y'$ the result follows.
\end{proof}

We then immediately have the following converse to
Proposition~\ref{PROP:patience-relationship-with-draw}.
\begin{corollary}
\label{COR:patience-relationship-with-draw}
\[
\patwld(n) \leq \patwl(2n) \enspace , \quad
\widehat{\patwld}(n) \leq \widehat{\patwl}(2n)\enspace .
\]
\end{corollary}

\subsubsection*{Patience Upper Bound}

Let $A$ be a $n \times n$ matrix game with integer entries. We shall
make use of Equation~(\ref{EQ:ShapleySnow})\footnote{Alternatively one
  could do essentially the same derivation using the standard
  formulation of matrix games as linear programs.}. Let $B$ be a $m
\times m$ submatrix of $A$ corresponding to an optimal strategy $x$ of
Player 1. Define the auxiliary $(m+1) \times (m+1)$ matrix $M
=\begin{bmatrix}0 & \allones^\transpose\\\allones & B
\end{bmatrix}$. Computing the determinant of $M$ by expanding along
first column and then the first row we find that $\det(M) =
-\allones^\transpose \adj(B) \allones$. Since the entries of $\adj(B)$
are integers, by Equation~(\ref{EQ:ShapleySnow}) we have that either
$x_i=0$ or $x_i \geq 1/\abs{\det(M)}$. We may thus bound the patience
of $x$ by $\abs{\det(M)}$.

Now, in case $A$ is a $(0,1)$ matrix game, the matrix $M$ is a $(0,1)$
matrix as well, of dimension at most $(n+1) \times (n+1)$. A bound of
Faddeev and Sominskii~\cite{FaddeevSominskii} then gives
$\abs{\det(M)} \leq (n+2)^\frac{n+2}{2}/2^{n+1}$.

Similarly, in case $A$ is a $(-1,0,1)$ matrix game, the matrix $M$ is
a $(-1,0,1)$ matrix as well, of dimension at most $(n+1) \times
(n+1)$, and using the Hadamard bound we get $\abs{\det(B)} \leq
(n+1)^\frac{n+1}{2}$. Combining these, the proof of Proposition~\ref{PROP:PatienceUpperbound} follows.

\section{Patience and Ill-conditioned Matrices}
\label{SEC:Patience-and-Ill-conditioned-Matrices}

From Theorem~\ref{THM:NonsigularMatrixGame} we see that a
nonsingular $n \times n$ matrix $A$ with $\allones^\transpose A^{-1}
\allones \neq 0$ defines a matrix game of patience at least
$\allones^\transpose A^{-1} \allones$, provided that \emph{both}
$\allones^\transpose A^{-1}/\allones^\transpose A^{-1}\allones > 0$
and $A^{-1}\allones/\allones^\transpose A^{-1}\allones > 0$.

For a non-singular $n \times n$ matrix $A$, let $B=A^{-1}=(b_{ij})$
and define $\chi(A)=\max_{i,j} |b_{ij}|$. The problem of constructing
$(0,1)$ or $(-1,1)$ matrices $A$ for which $\chi(A)$ is large was
considered first by Graham and Sloane~\cite{LAA:GS84}, and later by
Alon and V\~u~\cite{JCT:AV97}. Such matrices have besides the
direct application of constructing ill-conditioned matrices, several
applications such as flat simplices, coin weighing, indecomposable
hypergraphs, and weights of Boolean threshold
functions~\cite{LAA:GS84,SIDMA:Haastad94,JCT:AV97,PIT:Podolskii09,MFCS:BHPS10}.

Define $\chi_1(n)$ as the maximum of $\chi(A)$ over all non-singular
$n \times n$ $(0,1)$ matrices $A$. Define $\chi_2(n)$ to be the
analogous quantity where $(-1,1)$ matrices are considered instead.
Alon and V{\~u}~\cite{JCT:AV97}, building on the techniques of
H{\aa}stad, gave a near optimal construction of ill-conditioned
matrices.  More precisely they provide for every $n$ an explicit $n
\times n$ $(0,1)$ matrix $A_1$ and an explicit $n \times n$ $(-1,1)$
matrix $A_2$ such that $\chi(A_i) \geq n^{n/2}/2^{n(2+o(1))}$ for
$i=1,2$. When $n$ is a power of 2 these lower bounds may be improved
to $n^{n/2}/2^{n(1+o(1))}$. Upper bounds for $\chi_i(n)$ are derived
from the Hadamard inequality.
\begin{theorem}[Alon and V{\~u}]
\begin{gather*}
n^{\frac{n}{2}}/2^{n(2+o(1))} \leq \chi_i(n) \quad \text{for } i=1,2\\
\chi_1(n) \leq n^{\frac{n}{2}}/2^{n-1} \\
\chi_2(n) \leq (n-1)^{\frac{n-1}{2}}/2^{n-1}
\end{gather*}
\label{THM:AlonVu}
\end{theorem}

In their application to indecomposable hypergraphs, Alon and V\~u
construct a nonsingular $(0,1)$ $n\times n$ matrix $D$ such that
$y=D^{-1}\allones \geq 0$ and $\abs{y_1/y_2} \geq
n^{\frac{n}{2}}/2^{n(2+o(1))}$. Unfortunately this construction does
not ensure that also $\allones^\transpose D^{-1} \geq 0$, and hence we
cannot use it to give a matrix game of large patience as described
above.

It does however turn out that \emph{any} matrix $A$ with large
$\chi(A)$ can be used to construct a matrix game with patience
$\chi(A)$, as will be explained in the following section.

\subsection{The Matrix Switching Game}

Let $B$ be any $n \times n$ matrix. We call the operation of flipping
all the signs of an entire row a \emph{row switch}, and similarly the
operation of flipping all the signs of an entire column a \emph{column
  switch}. We are interested in the sum of all entries,
$\allones^\transpose B' \allones$, for matrices $B'$ obtained from $B$
using row and column switches. The matrix switching game for $B$ is to
find such a matrix $B'$ maximizing $\allones^\transpose B' \allones$,
the \emph{value} of the switching game. Equivalently we may view the
matrix switching game as the problem of maximizing the bilinear form
$x^\transpose B y$ over $x,y \in \{-1,1\}^n$. The special case of
matrix switching game for $(-1,1)$ matrices is known as the
Gale-Berlekamp switching game~\cite[Chapter 6]{Spencer87tenlectures}
(or simply the Berlekamp switching
game~\cite{Sloane87unsolvedproblems}).

Directly from the definition of the matrix switching game we have the
following.
\begin{lemma}
\label{LEM:LocalOptisGood}
  Let $B$ be any $n\times n$ matrix, such that $\allones^\transpose B
  \allones$ can not be increased by a row or column switch. Then
  $\allones^\transpose B \geq 0$ and $B \allones \geq 0$.
\end{lemma}

It is easy to see that the value of the matrix switching game is at
least as large as the largest element of the matrix.
\begin{lemma}
\label{LEM:LargestEntrySwitching}
  Let $B=(b_{ij})$ be any $n \times n$ matrix. Then there exist $x \in
  \{-1,1\}^n$ and $y \in \{-1,1\}^n$ such that $x^\transpose B y \geq
\max_{ij} \abs{b_{ij}}$.
\end{lemma}
\begin{proof}
  Let $b_{ij}$ be the entry of largest absolute value in $B$. First
  perform column switches in $B$ to make all entries of row $i$
  non-negative. Next perform a row switch in any row where the sum of
  the entries of the row is negative.

  Thus the value of the game is at least $\max_{ij} \abs{b_{ij}}$.
\end{proof}

\begin{proposition}
\label{PROP:SwitchMatrixGame01}
Let $A$ be a nonsingular $(-1,1)$ $n \times n$ matrix. Then there
exist a block checkerboard sign pattern $\Sigma$ such that the
$(n+1)\times(n+1)$ $(-1,0,1)$ matrix game $B$ given by
\[
B=\begin{bmatrix}
1 & 0\\
0 & A \circ \Sigma^\transpose
\end{bmatrix}
\]
satisfies $\tau^1(B) \geq \chi(A)$ and $\tau^2(B) \geq \chi(A)$.
\end{proposition}
\begin{proof}
  Let $x,y \in \{-1,1\}^n$ maximize the bilinear form $x^\transpose
  A^{-1} y$, and define $\Sigma = xy^\transpose$. Let $C = A \circ
  \Sigma^\transpose$. Then $C^{-1} = A^{-1} \circ \Sigma$.  By
  Lemma~\ref{LEM:LocalOptisGood} we have $\allones^\transpose C^{-1}
  \geq 0$ and $C^{-1}\allones \geq 0$, and by
  Lemma~\ref{LEM:LargestEntrySwitching} we have $\allones^\transpose
  C^{-1} \allones \geq \chi(A) > 0$. Hence
  Theorem~\ref{THM:NonsigularMatrixGame} gives $v := \val(C) =
  1/\allones^\transpose C^{-1} \allones > 0$.

  Optimal strategies in $B$ for Player 1 must then be of the form
  $(v/(1+v),x'^\transpose/(1+v))$, and similarly optimal strategies in
  $B$ for Player 2 must be of the form
  $(v/(1+v),y'^\transpose/(1+v))$, where $x$ and $y$ are optimal
  strategies in $C$ for Player 1 and Player 2. It follows these are of
  patience at least $(1+v)/v \geq 1/v = \allones^\transpose C^{-1}
  \allones \geq \chi(A)$.
\end{proof}

Combining this with Theorem~\ref{THM:AlonVu} gives the proof of
Theorem~\ref{THM:IndirectPatienceLowerbound} and
Corollary~\ref{COR:IndirectPatienceLowerbound}.

While these bounds matches the patience upper bound of
Proposition~\ref{PROP:PatienceUpperbound} up to the constant in the
exponent, and in the case of win-lose-draw games in fact match up to
an exponential factor, a drawback is that the matrices giving these
lower bounds are not very \emph{explicit}. We can formalize such a
statement using computational complexity theory, by looking at the
complexity of the computational task to compute the $n \times n$
matrix of the family, given as input the number $n$.

The matrices constructed by Alon and V\~u from
Theorem~\ref{THM:AlonVu} are in fact very explicit in this sense as
will be detailed in Section~\ref{SEC:AlonVuExplicitness}. Given such a
matrix $A$, the proof of Proposition~\ref{PROP:SwitchMatrixGame01}
proceeds to invert $A$, and then solve the matrix switching game for
$A^{-1}$. While inverting $A$ is a polynomial time computation (in
$n$) it turns out that solving the Matrix Switching game is
\NP-hard~\cite{TIT:RV08} (this is true even in the case of the
Gale-Berlekamp game). This \NP-hardness result is not a real obstacle
though; all that is needed for the proof of
Proposition~\ref{PROP:SwitchMatrixGame01} to go through is that the
vectors $x,y \in \{-1,1\}^n$ describe a \emph{local maximum} of the
matrix switching game, in the sense that the bilinear form
$x^\transpose A^{-1} y$ cannot be increased by changing a single
coordinate of $x$ or $y$, and furthermore that the value of this local
maximum is at least $\chi(A)$. We will discuss this issue of
explicitness in more detail in the next section.




\subsection{Local Search and Bipartite Maximum Cut}

In this section we will consider local search algorithms from the
perspective of computational complexity in order to properly discuss
the explicitness of the matrix games of high patience constructed in
the previous section.

Johnson~et.al.~\cite{JCSS:JPY88} formalized the notion of polynomial
time local search problems by the complexity class \PLS. A \emph{local
  search problem} $P$ is specified by the following:
\begin{itemize}
\item A set $\calI$ of \emph{instances}, given by a polynomial time algorithm
  that decides if a given string represents an instance of $P$.
\item For each instance $x$, a set $\calF(x)$ of \emph{feasible
    solutions}, whose elements are strings bounded polynomially in the
  length of $x$.
\item A specification of $P$ as either a maximization problem or a
  minimization problem, together with a cost measure $c(x,y)$, for
  $x\in\calI$ and $y \in \calF(x)$, to be maximized or
  minimized, respectively.
\item For each solution $y \in \calF(x)$, a set
  $\calN(x,y)$ of \emph{neighboring solutions}.
\end{itemize}
We then say that $P$ is in \PLS\ if there exist polynomial time
algorithms $A$, $M$, and $C$ as follows:
\begin{itemize}
\item $A$, on input $x \in \calI$, produces a solution $y \in \calF(x)$.
\item $M$, on input $x \in \calI$ and $y \in \calF(x)$,  computes $c(x,y)$.
\item $C$, on input $x \in \calI$ and $y \in \calF(x)$, either reports
  that $y$ is the best solution in $\calN(x,y)$, or produces a better
  solution $y' \in \calN(x,y)$.
\end{itemize}

The \emph{standard algorithm} to solve the problem $P$ is to first run
algorithm $A$, and then repeatedly run algorithm $C$ until a locally
optimum is found.

We can cast the matrix switching game of the previous section in this
framework, showing that the problem is in \PLS: Instances are integer
$n\times n$ matrices $A$. A solution is given by a pair of vectors
$x,y \in \{-1,1\}^n$. The cost of a solution is $x^\transpose A y$,
and neighbors of $(x,y)$ are those obtained by flipping the sign of an
entry in either $x$ or $y$. Algorithms $A$, $M$, and $C$ are
immediate.

We note that to implement the proof of
Proposition~\ref{PROP:SwitchMatrixGame01} one would need not only a
locally optimum, but a locally optimum of a certain quality. We will
suggest that it might be hard to even just find a locally optimum, or
in other words solve the matrix switching game, without using specific
knowledge of the input matrix $A$.

Using the notions of \emph{reductions} and \emph{completeness} one may
in a similar way to the theory of \NP-completeness argue that certain
problem in \PLS\ are unlikely to be solvable in polynomial time.

A \PLS-problem $P_1$ is \PLS-reducible~\cite{JCSS:JPY88} to another
\PLS-problem $P_2$, if there are polynomial time computable functions
$f$ and $g$, such that $f$ maps instances of $P_1$ to instances of
$P_2$, $g$ maps pairs $(y',x)$, where $y'$ is a solution of $f(x)$, to a
solution $g(y',x)$ of $x$, and in the case when $y'$ is a local optimum
of $P_2$, then $g(y,x)$ is a local optimum of $P_1$. With the notion
of reduction in place, we say that a \PLS-problem $P$ is
\PLS-complete, if every other \PLS-problem reduces to $P$ by a
\PLS-reduction.

Sch{\"a}ffer and Yannakakis~\cite{SICOMP:SY91} found a number of
natural local search problems to be \PLS-complete. In particular they
showed that the \MAXCUT\ problem is \PLS-complete under the \FLIP\
neighborhood. Here the \MAXCUT\ problem is given as follows: Instances
are graphs $G=(V,E)$, $V=\{1,\dots,n\}$, with integer weights on the
edges, $w_{ij}$. Solutions are cuts $(S,\overline{S})$, $S \subseteq
V$, of the vertices, and the cost of a solution is the sum of the
weights of edges connecting vertices across the cut, $\sum_{ij \in
  (S,\overline{S})} w_{ij}$. The \FLIP\ neighborhood is defined by the
action of moving a single vertex across the cut.

Letting $x_i \in \{-1,1\}$ be given by $x_i=1$ if and only if $i \in
S$, we see the \MAXCUT\ problem is equivalent to maximizing the
quadratic form $\sum_{i<j} w_{ij} (1-x_ix_j)/2$ over $x \in
\{-1,1\}^n$. Conversely, the problem of maximizing a quadratic form
$x^\transpose A x$ over $x \in \{-1,1\}^n$ can be formulated as a
\MAXCUT\ instance.

Since, as we have seen, the matrix switching game is equivalent to
maximizing a \emph{bilinear} form over $\{-1,1\}^n$, it is not
surprising that we can also reformulate the matrix switching game as a
maximum cut problem.

Indeed, let $A$ be any $n\times n$ matrix $A$, and $x,y \in
\{-1,1\}^n$. Then
\begin{equation}
\label{EQ:BipartiteMAXCUT}
(x^\transpose A y -\allones^\transpose A \allones)/2 = \sum_{i,j} -
a_{ij} (1-x_iy_j)/2 \enspace .
\end{equation}
Define now the bipartite graph $G_A=(V_1,V_2,E)$, with
$V_1=\{1,\dots,n\}$ and $V_2=\{1',\dots,n'\}$ , with an edge $(i,j)$
whenever $a_{ij}\neq 0$ of weight $w_{ij}=-a_{ij}$. This forms an
instance of the \emph{bipartite} \MAXCUT\ problem, which is the
restriction of the \MAXCUT\ problem to bipartite graphs. That is,
solutions are cuts $(S,\overline{S})$, $S \subseteq V_1 \cup V_2$, of
the vertices, and the cost of a solution is again the sum of the
weights of edges connecting vertices across the cut, $\sum_{ij \in
  (S,\overline{S})} w_{ij}$. Letting $x_i=1$ if and only if $i \in S$
and $y_i=1$ if and only if $i' \in S$, this is exactly the quantity
expressed in Equation~(\ref{EQ:BipartiteMAXCUT}).

The bipartite \MAXCUT\ problem is
\NP-hard~\cite{MP:McCRR03,TIT:RV08}. But as argued, the question
relevant for us should be if the bipartite \MAXCUT\ problem is
\PLS-complete under the \FLIP\ neighborhood as well (observe that the
\FLIP\ neighborhood corresponds to row and column switches in the
corresponding matrix switching game).

Despite our efforts, we have not been able to resolve this
question. On the other hand we have not been able to solve the problem
in polynomial time either. Failing to prove the problem to be
\PLS-complete does not \textit{a priori} suggest that the local search
problem might be easy. Actually, unlike the theory of \NP-completeness,
it is not rare to find local search problems that are neither known to
be polynomial time solvable, nor \PLS-complete. An important such
example is the \TSP\ problem. Krentel showed the problem to be
\PLS-complete under the $k$-\OPT\ neighborhood (which is defined by
the process of removing from a tour arbitrary $k$ edges and
reconnecting the $k$ resulting pieces into a new tour), for
sufficiently large $k$~\cite{FOCS:Krentel89}. It is still an open
problem if the same is true for the simple $2$-\OPT\ and $3$-\OPT\
neighborhoods~\cite{SICOMP:SY91}.

We are able to show a weaker statement about the bipartite \MAXCUT\
problem. Define the 2-\FLIP\ neighborhood by the action of moving up
to 2 vertices across the cut. We then have the following hardness result.

\begin{proposition}
  The bipartite \MAXCUT\ problem is \PLS-complete under the 2-\FLIP\
  neighborhood.
\end{proposition}
\begin{proof}
  The result is proved by a simple reduction from the ordinary
  \MAXCUT\ problem. Let $G=(V,E)$ with $V=\{1,\dots,n\}$ be a graph
  with weight $w$ forming a \MAXCUT\ instance. Let $M=1+\sum_{ij}
  \abs{w_{ij}}$. Define a bipartite graph $G'=(V_1,V_2,E')$ with weights
  $w'_{ij'}$ as follows. We let $V_1=\{1,\dots,n\}$ and
  $V_2=\{1',\dots,n'\}$. Vertices $i$ and $i'$ are joined by an edge
  of ``huge'' weight $M$. Whenever $ij \in E$ we join $i$ and $j'$
  as well as $i'$ and $j$ by an edge of weight $-w_{ij}$. We thus let
  $f(G,w)=(G',w')$ be the first function of the reduction.

  First, observe if a given cut $(S',\overline{S'})$ of the vertices
  of $G'$ does not satisfy $i \in S'$ if and only if $i' \in
  \overline{S'}$, then the weight of the cut can be improved by moving
  either of $i$ or $i'$ to the other partition of the cut. We will
  thus in the following calculation assume that this is not the case
  for any $i$.  Now a cut $(S',\overline{S'})$ of $G'$ induces a cut
  $(S,\overline{S})$ of $G$ defined by $i \in S$ if and only if $i \in
  S'$.  We then have
\[
\begin{split}
  w'(S',\overline{S'}) & = nM + \sum_{i \in S', j' \in \overline{S'}}
  w^{\prime}_{ij'} + \sum_{i' \in S', j \in \overline{S'}} w^{\prime}_{i'j} = nM +
  2\sum_{i \in S', j' \in \overline{S'}} w^{\prime}_{ij'} \\ & = nM - 2\sum_{i
    \in S, j \in S} w_{ij} = (nM -2\sum_{ij} w_{ij}) + 2\sum_{i \in
    S, j \in \overline{S}} w_{ij}
\end{split}
\]

From this we see that $w(S,\overline{S}) = w'(S',\overline{S'})/2
+ (\sum_{ij} w_{ij} - nM/2)$.

We can thus simply define the last function $g$ of the reduction as
$g(S',\overline{S'}) = (S,\overline{S})$.
\end{proof}

\begin{remark}
  One can observe that the reduction above in fact satisfies the
  notion of being \emph{tight} as defined by Sch\"affer and
  Yannakakis~\cite{SICOMP:SY91}. We will not define the notion here,
  but just remark that it implies that the reduction gives a number of
  additional results besides showing \PLS-completeness. For instance
  it implies that the standard algorithm must take exponential time in
  the worst case. And this is true no matter how the neighbors are
  chosen in each step of the local search procedure.
\end{remark}

\section{Explicit Examples of Exponential Patience}
\label{SEC:Exlicit-Exponential-Exaxmples}

In this section we give several examples of matrix games that requires
exponential patience. All of the examples are special \emph{Toeplitz
  matrices}, that were constructed earlier for the purpose of studying
ill-conditioned matrices \cite{LAA:GS84} and extremal matrices with
respect to the determinant \cite{LAA:Ching93}. Here a $n\times n$
matrix $A=(a_{ij})$ is called a Toeplitz matrix, if every
left-to-right descending diagonal of $A$ is constant, i.e.\
$a_{ij}=a_{i+1,j+1}$ for all $i$ and $j$. We may thus specify $A$ by
the $2n-1$ numbers $a_{n,1},\dots,a_{1,1},\dots,a_{1,n}$. We shall use
the notation
\[
A=\toep(a_{n,1}\dots a_{2,1} \underline{a_{1,1}} a_{1,2} \dots
a_{1,n})
\]
with the upper left element of $A$ underlined.

For all the examples we use Theorem~\ref{THM:NonsigularMatrixGame} to
compute the patience. For the first matrix we can do this directly,
whereas for the last two examples we need to go via
Lemma~\ref{LEM:Checkerboard}, using that the inverse of the matrices
turn out to weakly obey block checkerboard sign patterns. This in turn
gives rise to win-lose-draw matrix games, whereas the first matrix
give rise to a win-lose matrix game. These $n\times n$ win-lose-draw
matrix games can be converted to $2n \times 2n$ win-lose matrix games
of the same patience using
Proposition~\ref{PROP:win-lose-draw-to-win-lose}, but we can do better
in these examples using the tight connection between $n \times n$
$(0,1)$ matrices and $(n+1)\times(n+1)$ $(-1,1)$ matrices, we describe
next. Let $A$ be a $n\times n$ $(0,1)$ matrix. Then define a
$(n+1)\times(n+1)$ $(-1,1)$ matrix $\Phi(A)$ by
\begin{equation}
\label{EQ:PhiConversion}
\Phi(A) =
\begin{bmatrix}
1 & \allones^\transpose\\
-\allones & 2A-\allones\allones^\transpose
\end{bmatrix} \enspace .
\end{equation}
If $A$ is invertible then $\Phi(A)$ is invertible as well and we have
(cf.\ \cite{JCT:AV97})
\[
\Phi(A)^{-1} = \frac{1}{2}
\begin{bmatrix}
2-\allones^\transpose A^{-1}\allones & -\allones^\transpose A^{-1}\\
A^{-1}\allones & A^{-1}
\end{bmatrix} \enspace .
\]

\subsection{\texorpdfstring{$(0,1)$}{(0,1)} Hessenberg matrix}

Ching~\cite{LAA:Ching93} considered the following Hessenberg-Toeplitz
matrices. For given $n$ define the $n \times n$ matrix $D_n =
(d_{ij})$ by $d_{i,i-k}=1$ if $k \in \{-1,0,2,4,\ldots\}$ and
$d_{i,i-k}=0$ otherwise. Alternatively, $D_n =
\toep(0101\dots10\underline{1}10000 \dots)$.

It was shown by Ching that for any $n \times n$ $(0,1)$ Hessenberg
matrix $A_n$ (i.e.\ $A_n$ is a triangular matrix except that the
diagonals above and below the main diagonal may also be nonzero) with
$n>2$ satisfies $\abs{\det(A_n)} \leq \det(D_n)$. We remark however
that the matrix $D_n$ is actually the transpose of the upper-right $n
\times n$ submatrix of the matrix $T_{n+1}$, defined by Graham and
Sloane, that we will consider in the next subsection. Also, in fact
Graham and Sloane already obtained the result of Ching while showing
properties of $T_n$ (see \cite{LAA:GS84}, Lemma 9).

For us, however, the matrix $D_n$ has the advantage over the other
examples we consider, that it directly gives a matrix game where the
optimal strategies are totally mixed.

\begin{example}
\[
D_5 = \begin{bmatrix}
1&1&0&0&0\\0&1&1&0&0\\1&0&1&1&0\\0&1&0&1&1\\1&0&1&0&1
\end{bmatrix}
\]
\end{example}

Let $F_n$ denote the $n$th Fibonacci number, given by $F_n = F_{n-1} +
F_{n-2}$ for $n>2$, $F_1=F_2=1$, and $F_0=0$. Alternatively
\[
F_n = \left(\varphi^n - (1-\varphi)^n\right)/\sqrt{5} \enspace ,
\]
where $\varphi = (1+\sqrt{5})/2 = 1.61803...$ is the golden ratio.

We shall not compute the inverse of $D_n$, but just determine the
information needed to apply Theorem~\ref{THM:NonsigularMatrixGame}. We
can determine a recurrence for the determinant of $D_n$ by expanding
along the first row, and obtain $\det(D_n) = \det(D_{n-1}) +
\det(D_{n-2})$ for $n>2$. Also $\det(D_1) = \det(D_2) = 1$, and hence
$\det(D_n) = F_n$. Similarly, one may easily verify the following.
\begin{lemma}
  Let $\widetilde{x},\widetilde{y} \in \RR^n$ be defined by
  $\widetilde{x}_i = \widetilde{y}_{n-i+1} = F_i/F_n$ for $i < n$ and
  $\widetilde{x}_n = \widetilde{y}_1 = F_{n-2}/F_n$.  Then
  $\widetilde{x}^\transpose D_n = \allones^\transpose$ and $D_n
  \widetilde{y} = \allones$. Also $\sum_{i=1}^n \widetilde{x}_i =
  \sum_{i=1}^n \widetilde{y}_i = (2F_n -1)/F_n$, and hence
  $\allones^\transpose D^{-1}_n \allones = (2F_n -1)/F_n$.
\end{lemma}
From this and Theorem~\ref{THM:NonsigularMatrixGame} we immediately
obtain a statement about the matrix $D_n$ viewed as a matrix game.
\begin{proposition}
  The matrix game $D_n$ has value $v=F_n/(2F_n-1)$ and unique optimal
  strategies $x$ and $y$ for the two players, where $x_i = y_{n-i+1} =
  F_i/F_n$ for $i < n$ and $x_n = y_1 = vF_{n-2}/F_n$. In particular
  we have $x_1 = y_n = 1/(2F_n-1)$, and the patience of both $x$ and
  $y$ is precisely $2F_n-1$, and asymptotically $\Omega(\varphi^n)$.
\end{proposition}

\subsection{Triangular matrix}

Graham and Sloane~\cite{LAA:GS84} defined the following triangular
matrices. For given $n$ define the $n \times n$ matrix $t_n =
(t_{ij})$ by $t_{i,i+k}=1$ if $k \in \{0,1,3,5,\dots\}$ and
$t_{i,i-k}=0$ otherwise. Alternatively $t_n =
\toep(0\dots0\underline{1}101010\dots)$.
\begin{example}
\[
t_6 = \begin{bmatrix}
1&1&0&1&0&1\\0&1&1&0&1&0\\0&0&1&1&0&1\\0&0&0&1&1&0\\0&0&0&0&1&1\\0&0&0&0&0&1
\end{bmatrix}
\]
\end{example}

As pointed out by Graham and Sloane, the inverse of $t_n$ is again a
triangular Toeplitz matrix, namely
\[
t_n^{-1} = \toep(0,\dots,0,\underline{1},-F_1,F_2,-F_3,\dots,(-1)^{n-1}F_{n-1}) \enspace .
\]
This means that $t_n^{-1}$ weakly obeys the (block) checkerboard sign
pattern $\Sigma_n=(\sigma_{ij})$ given by $\sigma_{ij} =
(-1)^{i+j}$. We thus have that the matrix $\overline{t}_n = t_n \circ
\Sigma = \toep(0\dots0,\underline{1},-1,0,-1,0,-1,0\dots)$ has inverse
$\toep(0,\dots,0,\underline{1},F_1,F_2,\dots,F_{n-1})$.

One may now easily verify the following.
\begin{lemma}
  Let $\widetilde{x},\widetilde{y} \in \RR^n$ be defined by
  $\widetilde{x}_i = \widetilde{y}_{n-i+1} = F_{i+1}$.  Then
  $\widetilde{x}^\transpose \overline{t}_n = \allones^\transpose$ and
  $\overline{t}_n \widetilde{y} = \allones$.  Also $\sum_{i=1}^n
  \widetilde{x}_i = \sum_{i=1}^n \widetilde{y}_i = \allones^\transpose
  \overline{t}^{-1}_n \allones = F_{n+3}-2$.
\end{lemma}
And from Theorem~\ref{THM:NonsigularMatrixGame} we obtain the following.
\begin{proposition}
  The $(-1,0,1)$ matrix game $\overline{t}_n$ has value $v=1/(F_{n+3}-2)$ and
  unique optimal strategies $x$ and $y$ for the two players, where
  $x_i = y_{n-i+1} = F_{i+1}/(F_{n+3}-2)$. In particular we have $x_1
  = y_n = 1/(F_{n+3}-2)$, and the patience of both $x$ and $y$ is
  precisely $F_{n+3}-2$, and asymptotically $\Omega(\varphi^n)$.
\end{proposition}

We next derive a win-lose matrix game of similar patience using
Equation~(\ref{EQ:PhiConversion}). Define a vector $f$ and its reverse
$f_R$ by
\[
\begin{split}
f&=((-1)^{n-2}F_{n-2},\dots,-F_1,F_0,-1)^\transpose\\
f_R&=(-1,F_0,-F_1,\dots,(-1)^{n-2}F_{n-2})^\transpose \enspace .
\end{split}
\]
One may verify that the inverse of $\Phi(t_n)$ is then given by
\[
\Phi(t_n)^{-1} = \frac{1}{2} \begin{bmatrix}
(-1)^{n-2}F_{n-3} & f_R^\transpose\\
-f & t_n^{-1}
\end{bmatrix} \enspace ,
\]
and we see that $\Phi(t_n)^{-1}$ weakly obeys the similar (block)
checkerboard sign pattern $-\Sigma_{n+1}$. Thus the matrix
$\overline{t'}_n=\Phi(t_n)^{-1} \circ (-\Sigma_{n+1})$ has a non-negative inverse.

One may now easily verify the following.
\begin{lemma}
  Let $\widetilde{x},\widetilde{y} \in \RR^{n+1}$ be defined by
  $\widetilde{x}_1 = \widetilde{y}_{n+1} = F_{n-1}$, and
  $\widetilde{x}_i = \widetilde{y}_{n-i+2} = F_{i-1}$, for $i\geq2$.
  Then $\widetilde{x}^\transpose \overline{t'}_n =
  \allones^\transpose$ and $\overline{t'}_n \widetilde{y} = \allones$.
  Also $\sum_{i=1}^{n+1} \widetilde{x}_i = \sum_{i=1}^{n+1} \widetilde{y}_i =
  \allones^\transpose \overline{t'}^{-1}_n \allones = F_{n-1}+F_{n+2}-1$.
\end{lemma}
And from Theorem~\ref{THM:NonsigularMatrixGame} we obtain the following.
\begin{proposition}
  The $(-1,1)$ matrix game $\overline{t'}_n$ has value
  $v=1/(F_{n-1}+F_{n+2}-1)$ and unique optimal strategies $x$ and $y$
  for the two players, where $x_1=y_{n+1}=F_{n-1}/(F_{n-1}+F_{n+2}-1)$ and
  $x_i = y_{n-i+2} = F_{i-1}/(F_{n-1}+F_{n+2}-1)$, for $i\geq2$

  In particular we have $x_2 = y_{n+1} = 1/(F_{n-1}+F_{n+2}-1)$, and
  the patience of both $x$ and $y$ is precisely $F_{n-1}+F_{n+2}-1$,
  and asymptotically $\Omega(\varphi^n)$.
\end{proposition}

\subsection{Toeplitz matrix}

Graham and Sloane~\cite{LAA:GS84} additionally defined the following
Toeplitz matrices. For given $n$ define the $n \times n$ matrix $T_n =
(t_{ij})$ by $t_{i,i-k}=1$ if $k \in \{-3,-1,0,3,4,6,7,\ldots\}$ and
$t_{i,i-k}=0$ otherwise. Alternatively $T_n =
\toep(\dots11001100\underline{1}10100000\dots)$. The inverse of $T_n$
was computed by Graham and Sloane using Trench's algorithm, and is
described below.

Define sequences $\{p_n\}$ and $\{q_n\}$ of integers by $q_0=0$,
$p_0=q_1=q_2=1$, and inductively
\begin{align*}
p_n & = p_{n-1} + q_{n-1} \enspace , \quad \text{for } n \geq 1 \enspace ,\\
q_n & = q_{n-1} + p_{n-2} \enspace , \quad \text{for } n \geq 3 \enspace .
\end{align*}
From these definitions one may determine their asymptotics as $p_n,
q_n = \Omega(\rho_3^n)$, where $\rho_3=1.75488...$ is the largest root
of $x^3-2x^2+x-1$.

{
\renewcommand\thefootnote{\fnsymbol{footnote}}
It turns out that the inverse of $T_n$ is symmetric about the
top-right to bottom-left diagonal and of the following form.
\[
T_n^{-1}= \left[
\begin{array}{rrrrrr|rrr}
-1 & -p_2 & -p_3 & \dots & -p_{n-4} & -p_{n-3} &  q_{n-2} & -q_{n-3}\footnotemark &  p_{n-2}\\
 1 &  q_1 &  q_2 & \dots &  q_{n-5} &  q_{n-4} & -p_{n-4} &  p_{n-5} & -q_{n-3}
\addtocounter{footnote}{-1}\addtocounter{Hfootnote}{-1}\footnotemark\\
-1 & -q_2 & -q_3 & \dots & -q_{n-4} & -q_{n-3} &  p_{n-3} & -p_{n-4} &  q_{n-2}\\\hline
 1 &  p_1 &  p_2 & \dots &  p_{n-5} &  p_{n-4} & -q_{n-3} &  q_{n-4} & -p_{n-3}\\
 0 &    1 &  p_1 & \dots &  p_{n-6} &  p_{n-5} & -q_{n-4} &  q_{n-5} & -p_{n-4}\\
 0 &    0 &    1 & \dots &  p_{n-7} &  p_{n-6} & -q_{n-5} &  q_{n-6} & -p_{n-5}\\
 \vdots & \vdots & \ddots & \ddots & \ddots & \vdots & \vdots & \vdots & \vdots\\
\end{array}
\right]
\]
\footnotetext{In \cite{LAA:GS84} this entry is incorrectly written as $-q_{n-1}$.}
}
Note that $T_n^{-1}$ weakly obeys the block checkerboard sign pattern
\[
\Sigma =
(-1,1,-1,1,\dots,1) (1,\dots,1,-1,1,-1)^\transpose \enspace .
\]
By Lemma \ref{LEM:Checkerboard} the matrix $\overline{T}_n = T_n \circ
\Sigma^\transpose$ has as inverse a matrix with all entries being
non-negative. We only give an asymptotic lower bound on the patience
of $\overline{T}_n$.

\begin{proposition}
  The $(-1,0,1)$ matrix game $\overline{T}_n$ has unique optimal
  strategies $x$ and $y$ for the two players each of patience
  $\Omega(\rho_3^n)$.
\end{proposition}
\begin{proof}
  Note that $\allones^\transpose \overline{T}_n^{-1}\allones =
  \Omega(p_n) = \Omega(\rho_3^n)$, $(\allones^\transpose
  \overline{T}_n^{-1})_1 = (\overline{T}_n^{-1}\allones)_n =
  4$. Using Theorem~\ref{THM:NonsigularMatrixGame} we have that $x_1 =
  y_n = 1/\Omega(\rho_3^n)$, and the patience of $x$ and $y$ is
  $\Omega(\rho_3^n)$.
\end{proof}

Again we may derive a win-lose matrix game of similar patience using
Equation~(\ref{EQ:PhiConversion}). Define a vector $g$ and its reverse
$g_R$ by
\[
\begin{split}
g&=(1,0,1,0,\dots,0)^\transpose\\
g_R&=(0,\dots,0,1,0,1)^\transpose \enspace .
\end{split}
\]
One may verify that the inverse of $\Phi(T_n)$ is then given by
\[
\Phi(T_n)^{-1} = \frac{1}{2} \begin{bmatrix}
0 & -g_R^\transpose\\
g & T_n^{-1}
\end{bmatrix} \enspace ,
\]
and we see that $\Phi(T_n)^{-1}$ weakly obeys the similar (block)
checkerboard sign pattern
\[
\Sigma' =
(1,-1,1,-1,1,\dots,1) (-1,1,\dots,1,-1,1,-1)^\transpose \enspace .
\]
Thus the matrix $\overline{T'}_n=\Phi(T_n)^{-1} \circ \Sigma'$ has a
non-negative inverse. We then have the following.
\begin{proposition}
  The $(-1,1)$ matrix game $\overline{T'}_n$ has unique optimal strategies $x$
  and $y$ for the two players each of patience $\Omega(\rho_3^n)$.
\end{proposition}
\begin{proof}
  Note that $\allones^\transpose \overline{T'}_n^{-1}\allones =
  \Omega(p_n) = \Omega(\rho_3^n)$, $(\allones^\transpose
  \overline{T}_n^{-1})_1 = (\overline{T}_n^{-1}\allones)_1 =
  1$. Using Theorem~\ref{THM:NonsigularMatrixGame} we have that $x_1 =
  y_1 = 1/\Omega(\rho_3^n)$, and the patience of $x$ and $y$ is
  $\Omega(\rho_3^n)$.
\end{proof}

\section{Explicit matrices of almost worst case patience}
\label{SEC:Almost-Worst-case}

In this section we present the proof of our main result,
Theorem~\ref{THM:MainLowerBound}. The overall strategy for the proof
is similar to the last examples of
Section~\ref{SEC:Exlicit-Exponential-Exaxmples}. Namely for any
$n=2^m$, we first construct a non-singular $n \times n$ $(-1,1)$
matrix $A$ for which $\chi(A) \geq
n^{\frac{n}{2}}/2^{n(1+o(1))}$. This matrix is a specific instance of
the ill-conditioned matrices constructed by Alon and
V\~{u}~\cite{JCT:AV97}. This immediately means that the inverse of $A$
has an entry of magnitude $n^{\frac{n}{2}}/2^{n(1+o(1))}$ by the
analysis of Alon and V\~{u} (or alternatively it is easily derived
from the more involved analysis of this section). But just as
important for us, the specifics of our construction allows us to show
that $A^{-1}$ weakly obeys a block checkerboard sign pattern
$\Sigma$. Using Lemma~\ref{LEM:Checkerboard} this means that the
$(-1,1)$ matrix $B = (b_{ij}) = A \circ \Sigma^\transpose$ has a non-negative
inverse. We can then apply the result of Shapley and Snow to analyze
the patience of the matrix game $B$. Specifically by
Theorem~\ref{THM:NonsigularMatrixGame}, the matrix game $B$ has unique
optimal strategies $x$ and $y$. In particular the strategy $x$ is
given by
\[
x^\transpose = \allones^\transpose B^{-1}/\allones^\transpose B^{-1}
\allones \enspace .
\]
In Section~\ref{SEC:FirstColumn-sign-pattern} we compute the first
column of $A^{-1}$ and from this analysis we find
\[
\sum_{i=1}^n b_{i1} = -(1-m/2) + m/2 = m-1 \enspace .
\]
Since $B^{-1}$ is non-negative we also have $\allones^\transpose
B^{-1}\allones \geq \chi(A) \geq n^{\frac{n}{2}}/2^{n(1+o(1))}$. Thus the
patience of $B$ is at least $n^{\frac{n}{2}}/2^{n(1+o(1))}$ as well.

The rest of the section is organized as follows. In
Section~\ref{SEC:Alon-Vu} we review the details of the construction of
Alon and V\~{u}. In Section~\ref{SEC:Ordering} we define the specific
instance of this construction that we will use. In
Sections~\ref{SEC:SignPattern} and the following three sections we
show that the matrix weakly obeys a specific block checkerboard sign
pattern. Finally in Section~\ref{SEC:AlonVuExplicitness} we give a
sketch of the proof that our construction is fully explicit.

\subsection{The \texorpdfstring{Alon-V\~u}{Alon-Vu} matrix}
\label{SEC:Alon-Vu}
We review here the matrix construction of Alon and V\~u. Let $n = 2^m$
be a power of two. Let $\alpha_1,\dots,\alpha_n$ be an ordering of the
$n$ subsets of $[m]$ that satisfies
\begin{itemize}
\item $\abs{\alpha_i} \leq \abs{\alpha_{i+1}}$
\item $\abs{\alpha_i \symdiff \alpha_{i+1}} \leq 2$.
\end{itemize}
Such an ordering was shown to exist by
H{\aa}stad~\cite{SIDMA:Haastad94}. In Section~\ref{SEC:Ordering} we
will construct a particular such ordering. Define for convenience
$\alpha_0=\emptyset$. Given the ordering we define a $n \times n$
$(-1,1)$ matrix $A=(a_{ij})$ by the following rules
(for intuition behind this construction see~\cite{SIDMA:Haastad94}).
\begin{enumerate}
\item If $\alpha_j \cap (\alpha_{i-1} \cup \alpha_i) = \alpha_{i-1} \symdiff \alpha_i$
  and $\abs{\alpha_{i-1} \symdiff \alpha_i}=2$, then $a_{ij}=-1$.
\item $\alpha_j \cap (\alpha_{i-1} \cup \alpha_i) \neq \emptyset$, but
  case (1) does not occur, then
  $a_{ij}=(-1)^{\abs{\alpha_{i-1}\cup\alpha_j}+1}$.
\item If $\alpha_j \cap (\alpha_{i-1} \cup \alpha_i) = \emptyset$, then $a_{ij}=1$.
\end{enumerate}

For analyzing the matrix, Alon and V\~u essentially considered a LQ
decomposition of $A$. Define the $(-1,1)$ matrix $Q = (q_{ij})$ by
$q_{ij}=(-1)^{\abs{\alpha_i \cap \alpha_j}}$. Then $Q$ is a symmetric
Hadamard matrix, $Q^2 = nI$.

For $i>1$, define subsets $A_i$ of $[n]$ by $A_i = \alpha_{i-1} \cup
\alpha_i$, and from these, further define families $F_i$ of subsets of
$[n]$ by the following rules.
\begin{enumerate}
\item If $\abs{\alpha_{i-1} \symdiff \alpha_i}=2$, then $F_i =
  \{\alpha_s \mid \alpha_s \subseteq A_i, \abs{\alpha_s \cap
    (\alpha_{i-1} \symdiff \alpha_i)}=1\}$.
\item If $\abs{\alpha_{i-1} \symdiff \alpha_i}=1$, then $F_i =
  \{\alpha_s \subseteq A_i \}$
\end{enumerate}
Whenever $\abs{\alpha_i}=k$ we have $\abs{F_i}=2^k$ for both
cases. Next, define the lower triangular matrix $L=(l_{ij})$ as
follows. Let $l_{11}=1$ and $l_{1j}=0$, for $j>1$. For $i>1$ we let
\[
l_{ij} = \begin{cases}
(\frac{1}{2})^{k-1}-1 & \text{if } j=i-1\\
(\frac{1}{2})^{k-1}   & \text{if } \alpha_j \in F_i \setminus \{\alpha_{i-1}\}\\
0                    & \text{if } \alpha_j \notin F_i
\end{cases} \enspace .
\]
One can then verify the following.
\begin{lemma}[\cite{JCT:AV97}, Lemma 2.1.2]
\label{LEM:AlonVu}
\[
A=LQ \enspace .
\]
\end{lemma}

\subsection{The ordering.}
\label{SEC:Ordering}

Here we construct a specific ordering of the subsets of $[m]$
satisfying the requirements given in Section~\ref{SEC:Alon-Vu}. We
first construct separate orderings for the subsets of size $k$ for
every $k$. These will have the property that the first set in the
order is the lexicographically smallest set, i.e.\ $\{1,\dots,k\}$,
and the last set of the order is the lexicographically largest set,
i.e.\ $\{m-k+1,\dots,m\}$.

If $\beta \subseteq [m]$ denote by $(\beta+i)$ the subset of $[m+i]$
defined by $(\beta+i) = \{j+i \in [m+i] \mid j \in \beta\}$ (this
definition makes sense also when $\beta=\emptyset$, in which case the
result is also $\emptyset$).  Let
$\boldsymbol\beta=(\beta_1,\dots,\beta_\ell)$ be an ordering of the
subsets $\beta_1,\dots,\beta_\ell \subseteq [m]$. We denote by
$\rev(\boldsymbol\beta)$ the reverse ordering
$\rev(\boldsymbol\beta)=(\beta_\ell,\dots,\beta_1)$. By
$(\boldsymbol\beta + i)$ we denote the ordering $(\boldsymbol\beta +
i)=((\beta_1+i),\dots,(\beta_\ell+i))$. By $(\{i\} \cup
\boldsymbol\beta )$ we denote the ordering $(\{i\} \cup
\boldsymbol\beta )=((\{i\} \cup \beta_1 ),\dots,(\{i\} \cup
\beta_\ell))$. (These definitions make sense even if
$\boldsymbol\beta$ is the empty list, resulting in the empty list as
well). If $\boldsymbol\beta'=(\beta'_1,\dots,\beta'_{\ell'})$ is
another ordering of different subsets, we denote by
$\boldsymbol\beta\circ \boldsymbol\beta'$ the ordering
$\boldsymbol\beta\circ \boldsymbol\beta' =
(\beta_1,\dots,\beta_\ell,\beta'_1,\dots,\beta'_{\ell'})$.

The separate ordering for subsets of size $k$ of $[m]$, is defined by
induction on $k$ and $m$. Denote this ordering by
$\boldsymbol\beta^{(k)}_{m}$. For $k=0$ we have just the empty set
$\emptyset$, and hence $\boldsymbol\beta^{(0)}_m = (\emptyset)$. For
convenience, define $\boldsymbol\beta^{(k)}_m$ as the empty order
$\boldsymbol\beta^{(k)}_m = ()$, when $k>m$.

We now construct the ordering of subsets of size $k$ of $[m]$ for $m
\geq k > 0$. Assume by induction we have ordered the subsets of size
$k-1$ of $[m']$ for all $m' \geq k-1$. Then we define
\[
  \boldsymbol\beta^{(k)}_{m} = (\{1\} \cup (\boldsymbol\beta^{(k-1)}_{m-1}+1)) \circ
  (\{2\} \cup \rev(\boldsymbol\beta^{(k-1)}_{m-2}+2)) \circ
  (\boldsymbol\beta^{(k)}_{m-2}+2) \enspace .
\]
We see the ordering begins with the first $\binom{m-1}{k-1}$ sets
containing element $1$, starting with the set $\{1,\dots,k\}$ and
ending with the set $\{1,m-k+2,\dots,m\}$. Next follows the
$\binom{m-2}{k-1}$ sets containing element $2$ but not element $1$,
starting with the set $\{2,m-k+2,\dots,m\}$ and ending with the set
$\{2,\dots,k+1\}$. Finally follows all the $\binom{m-2}{k}$ sets not
containing the elements $1$ and $2$, starting with the set
$\{3,\dots,k+2\}$ and ending with the set $\{m-k+1,\dots,m\}$. We note
that between all these neighboring ending sets and starting sets the
symmetric difference is exactly 2, and we have covered all
$\binom{m}{k}=\binom{m-1}{k-1}+\binom{m-2}{k-1}+\binom{m-2}{k}$
sets. By induction the ordering thus satisfies the requirement of
symmetric differences being at most 2. Note that the first set is the
lexicographically smallest set, and the last set is the
lexicographically largest set.

\begin{example}
For $m=4$ we have the following orderings.
\[
\begin{split}
\boldsymbol\beta^{(0)}_{4} &= (\emptyset)\\
\boldsymbol\beta^{(1)}_{4} &= (\{1\},\{2\},\{3\},\{4\})\\
\boldsymbol\beta^{(2)}_{4} &= (\{1,2\},\{1,3\},\{1,4\},\{2,4\},\{2,3\},\{3,4\})\\
\boldsymbol\beta^{(3)}_{4} &= (\{1,2,3\},\{1,2,4\},\{1,3,4\},\{2,3,4\})\\
\boldsymbol\beta^{(0)}_{4} &= (\{1,2,3,4\})\\
\end{split}
\]
\end{example}

Next, we construct the full ordering $\boldsymbol\beta_m$ by combining
all $\boldsymbol\beta^{(k)}_m$. First, for $k=2$ we define a
\emph{shifted} version of the ordering. Let $\ell=\binom{m}{2}$. Let
$\boldsymbol\beta^{(2)}_m=(\beta^{(2)}_{m,1},\dots,\beta^{(2)}_{m,\ell})$. Then
define
$\widehat{\boldsymbol\beta^{(2)}_m}=(\widehat{\beta^{(2)}_{m,1}},\dots,\widehat{\beta^{(2)}_{m,\ell}})$
by $\widehat{\beta^{(2)}_{m,i}} = \{(j-1) \mod m \mid j \in
\beta^{(2)}_{m,i}\}$. Having this shifted version of sets of size 2
will be critically used in our proof. Now we can finally define
\[
\boldsymbol\beta_m = \boldsymbol\beta^{0}_m \circ
\boldsymbol\beta^{1}_m \circ \widehat{\boldsymbol\beta^{(2)}_m} \circ
\rev(\boldsymbol\beta^{3}_m) \circ \boldsymbol\beta^{4}_m \circ \rev(\boldsymbol\beta^{5}_m) \circ \boldsymbol\beta^{6}_m \circ \dots
\circ \boldsymbol\beta^{m}_m \enspace .
\]
In other words, $\boldsymbol\beta_m$ begins with the concatenation of
$\boldsymbol\beta^{0}_m$, $\boldsymbol\beta^{1}_m$, and
$\widehat{\boldsymbol\beta^{(2)}_m}$, after which the remaining orders
$\boldsymbol\beta^{k}_m$ are concatenated with
$\boldsymbol\beta^{k}_m$ reversed if $k$ is odd.

We now verify that the two properties the order must satisfy
holds. Clearly the sets are ordered in nondecreasing size. We have
already established the requirement about symmetric differences within
each order of sets of a given size. We next consider the pairs of
ending sets and starting sets. The first 3 such pairs are
$(\emptyset,\{1\})$, $(\{m\},\{1,m\})$, and
$(\{m-2,m-1\},\{m-2,m-1,m\})$. In each case the symmetric difference
is exactly $1$. The general case follows by recalling that each order
starts with the lexicographically smallest set and ends with the
lexicographically largest set, and every second order is reversed.

\begin{example}
\[
\begin{split}
\boldsymbol\beta_4 = (&\emptyset,\\
 &\{1\},\{2\},\{3\},\{4\}\\
 &\{4,1\},\{4,2\},\{4,3\},\{1,3\},\{1,2\},\{2,3\}\\
 &\{2,3,4\},\{1,3,4\},\{1,2,4\},\{1,2,3\}\\
 &\{1,2,3,4\})\\
\end{split}
\]
\end{example}

\subsection{Sign pattern of the inverse of \texorpdfstring{$A$}{A}}
\label{SEC:SignPattern}

In this section we let $\alpha_1,\dots,\alpha_n$ denote the particular
ordering defined in Section~\ref{SEC:Ordering}, and we consider the
construction of the matrix $A$ from Section~\ref{SEC:Alon-Vu} with
respect to this ordering, together with the corresponding matrices $L$
and $Q$, sets $A_i$, and families of subsets $F_i$.

\begin{definition}
  For a subset $\alpha \subseteq \{1,\dots,m\}$ we let $\num(\alpha)$
  denote the unique $j \in \{1,\dots,n\}$ such that $\alpha=\alpha_j$.
  Define also $i_k=\min \{j : \abs{\alpha_j}=k\}$, for all $k$.
\end{definition}
We remark that $i_k$ does not depend on the particular order we
consider, but is fully defined by the conditions of
Section~\ref{SEC:Alon-Vu}.

We prove that the matrix $A$, for $m \geq 2$ has an inverse that
weakly obeys a block checkerboard sign pattern. Namely we show that
$A^{-1}$ weakly obeys a sign pattern $\Sigma$ of the following kind.
\[
\Sigma=\left[ \begin{array}{rrr|rrrr}
     & -1 &  &  & 1 &  &  \\ \hline
     & 1 &  &  & -1  &  &  \\ \hline
     & -1 &  &  & 1  &  &  \\ \hline
     & 1 &  &  & -1  &  &  \\ \hline
     & \vdots &  &  & \vdots &  &
  \end{array} \right],
\]
The $n$ rows are divided into $m+1$ blocks. Block $k$ corresponds to
the subsets of size $k-1$, for $k=1,\dots,m+1$. That is, block $k$
consists of the rows $i$ for which $\abs{\alpha_i}=k-1$. The columns
are divided into precisely two blocks. For $m\geq 6$, we in fact prove
that the first block of columns is of size $2m-1$. Thus $\Sigma =
\sigma^{(1)}(\sigma^{(2)})^\transpose$, where
\[
\sigma^{(1)}_i = (-1)^{\abs{\alpha_i}} \quad \text{and} \quad \sigma^{(2)} = (\overbrace{-1,\dots,-1}^{2m-1},\overbrace{1,\dots,1}^{n-2m+1}) \enspace .
\]
One may verify by hand that the matrices for $m=4$ and $m=5$ also
weakly obeys this sign pattern. The matrices for $m=2$ and $m=3$ do
not weakly obey this sign pattern, but the similar sign pattern where
the columns are divided into two blocks of equal size. The matrix for
$m=1$ does not weakly obey a block checkerboard sign pattern.

We prove this for the first column as a special case in
Subsection~\ref{SEC:FirstColumn}. The remaining columns we handle as follows. By
Lemma~\ref{LEM:AlonVu} we have $A=LQ$. Column $j$ of $A^{-1}$ is then
the solution of the linear system
\begin{equation}
LQx = e_j
\end{equation}
Defining $z=Qx$, we have the equivalent system
\begin{equation}
Lz=e_j
\end{equation}
In Subsections~\ref{SEC:SecondBlock}, \ref{SEC:ThirdBlock}, and
\ref{SEC:RemainingColumns} we prove that for $j>1$ we have $\abs{z_n}
> \sum_{i=1}^{n-1} \abs{z_i}$. Then, since $x=\frac{1}{n}Qz$, by the
definition of $Q$ we have
\[
\sgn(x_i) = (-1)^{\abs{\alpha_i}}\sgn(z_n)
\]
Furthermore, we prove that $z_n>0$ for $i < 2m$, and $z_n<0$ for $i
\geq 2m$, thus establishing the claimed sign pattern. Note that this
latter part is not necessary to claim that $A^{-1}$ weakly obeys a
block checkerboard sign pattern. But proving this allows us to argue
that the construction is fully-explicit.

We will several times use the following simple facts.
\begin{lemma}
\label{LEM:fastgrowth}
  Let $w_1\geq 1$ and $w_{i+1} \geq (2+\epsilon)w_i$, for $i\geq
  1$ and $\epsilon>0$. Let $c>0$. Then
\[
w_s > \sum_{\ell=1}^{s-1} w_{\ell} + cw_1 \enspace ,
\]
for $s \geq \log_{2+\epsilon}(c/\epsilon)+2$.
\end{lemma}
\begin{proof}
Clearly $w_i > \sum_{\ell=1}^{i-1} w_\ell$ for all $i$, and thus
\[
w_s \geq (2+\epsilon)w_{s-1} > w_{s-1} + \sum_{\ell=1}^{s-2} w_\ell + \epsilon w_{s-1} \geq
\sum_{\ell=1}^{s-1} w_\ell + \epsilon(2+\epsilon)^{s-2}w_1 \geq \sum_{\ell=1}^{s-1} w_\ell + cw_1 \enspace .
\]
\end{proof}

\begin{lemma}
\label{LEM:fastgrowth2}
Let $z$ be the solution of $Lz=e_i$, and let $s \geq i_3$ be such that
$\abs{z_s} > \sum_{\ell=1}^{s-1} \abs{z_\ell}$.
Then $\abs{z_n} > \sum_{\ell=1}^{n-1} \abs{z_\ell}$,
and $\sgn(z_n) = \sgn(z_s)$.
\end{lemma}
\begin{proof}
  The proof is a simple induction argument. Let $j>s$, and
  $k=\abs{\alpha_j}$. If $\ell<i$ then $z_\ell=0$. Hence our assumption
  implies that $s\geq i$, and thus $j>i$. We then have the equation
\[
\frac{1}{2^{k-1}}\left( (1-2^{k-1})z_{j-1} + \sum_{\alpha_\ell \in F_j \setminus
    \{\alpha_{j-1}\}} z_{\alpha_\ell} \right) = 0 \enspace ,
\]
which means
\[
z_j = (2^{k-1}-1)z_{j-1} - \sum_{\alpha_\ell \in F_j \setminus
  \{\alpha_{j-1},\alpha_j\}} z_{\alpha_\ell} \enspace .
\]
We treat the case of $z_s > 0$; the case of $z_s<0$ is analogous.  By
induction we can estimate
\[
z_j \geq (2^{k-1}-1)z_{j-1} -
\sum_{\ell=1}^{j-2}\abs{z_\ell} \geq (2^{k-1}-2)z_{j-1} > \sum_{\ell=1}^{j-1} \abs{z_\ell} \enspace ,
\]
since $k\geq 3$.
\end{proof}

Below we state as an example the matrices $L$ (with zero entries omitted), $A$, and $A^{-1}$, for $m=4$.

\[
\setlength{\arraycolsep}{1pt}
L=\left[
\begin{array}{L|LLLL|LLLLLL|LLLL|L}
\usebox{\Lone}&&&&&&&&&&&&&&&\tabularnewline\hline
&\usebox{\Lone}&&&&&&&&&&&&&&\tabularnewline
&&\usebox{\Lone}&&&&&&&&&&&&&\tabularnewline
&&&\usebox{\Lone}&&&&&&&&&&&&\tabularnewline
&&&&\usebox{\Lone}&&&&&&&&&&&\tabularnewline\hline
\nicefrac12&\nicefrac12&&&\nicefrac{-1}2&\nicefrac12&&&&&&&&&&\tabularnewline
&\nicefrac12&\nicefrac12&&&\nicefrac{-1}2&\nicefrac12&&&&&&&&&\tabularnewline
&&\nicefrac12&\nicefrac12&&&\nicefrac{-1}2&\nicefrac12&&&&&&&&\tabularnewline
&\nicefrac12&&&\nicefrac12&&&\nicefrac{-1}2&\nicefrac12&&&&&&&\tabularnewline
&&\nicefrac12&\nicefrac12&&&&&\nicefrac{-1}2&\nicefrac12&&&&&&\tabularnewline
&\nicefrac12&&\nicefrac12&&&&&&\nicefrac{-1}2&\nicefrac12&&&&&\tabularnewline\hline
\nicefrac14&&\nicefrac14&\nicefrac14&\nicefrac14&&\nicefrac14&\nicefrac14&&&\nicefrac{-3}4&\nicefrac14&&&&\tabularnewline
&\nicefrac14&\nicefrac14&&&\nicefrac14&\nicefrac14&&\nicefrac14&&\nicefrac14&\nicefrac{-3}4&\nicefrac14&&&\tabularnewline
&&\nicefrac14&\nicefrac14&&&\nicefrac14&\nicefrac14&\nicefrac14&\nicefrac14&&&\nicefrac{-3}4&\nicefrac14&&\tabularnewline
&&&\nicefrac14&\nicefrac14&\nicefrac14&\nicefrac14&&\nicefrac14&&\nicefrac14&&&\nicefrac{-3}4&\nicefrac14&\tabularnewline\hline
\nicefrac18&\nicefrac18&\nicefrac18&\nicefrac18&\nicefrac18&\nicefrac18&\nicefrac18&\nicefrac18&\nicefrac18&\nicefrac18&\nicefrac18&\nicefrac18&\nicefrac18&\nicefrac18&\nicefrac{-7}8&\nicefrac18\tabularnewline
\end{array}
\right]
\]

\[
\setlength{\arraycolsep}{2pt}
A=\left[
\begin{array}{r|rrrr|rrrrrr|rrrr|r}
 1 &  1 &  1 &  1 &  1 &  1 &  1 &  1 &  1 &  1 &  1 &  1 &  1 &  1 &  1 &  1 \\\hline
 1 & -1 &  1 &  1 &  1 & -1 &  1 &  1 & -1 & -1 &  1 &  1 & -1 & -1 & -1 & -1 \\
 1 &  1 & -1 &  1 &  1 &  1 & -1 &  1 &  1 & -1 & -1 & -1 &  1 & -1 & -1 & -1 \\
 1 &  1 &  1 & -1 &  1 &  1 &  1 & -1 & -1 &  1 & -1 & -1 & -1 &  1 & -1 & -1 \\
 1 &  1 &  1 &  1 & -1 & -1 & -1 & -1 &  1 &  1 &  1 & -1 & -1 & -1 &  1 & -1 \\\hline
 1 & -1 &  1 &  1 &  1 &  1 &  1 &  1 & -1 & -1 &  1 &  1 &  1 &  1 & -1 &  1 \\
 1 &  1 & -1 &  1 &  1 & -1 &  1 &  1 &  1 & -1 & -1 &  1 & -1 & -1 & -1 & -1 \\
 1 &  1 &  1 & -1 &  1 &  1 & -1 &  1 & -1 &  1 & -1 & -1 &  1 & -1 & -1 & -1 \\
 1 & -1 &  1 &  1 &  1 & -1 &  1 & -1 &  1 & -1 &  1 & -1 & -1 & -1 &  1 & -1 \\
 1 &  1 & -1 &  1 &  1 &  1 & -1 &  1 & -1 &  1 & -1 & -1 & -1 &  1 & -1 & -1 \\
 1 &  1 &  1 & -1 &  1 &  1 &  1 & -1 & -1 & -1 &  1 &  1 & -1 & -1 & -1 & -1 \\\hline
 1 &  1 &  1 &  1 & -1 & -1 &  1 &  1 &  1 &  1 & -1 & -1 &  1 &  1 & -1 & -1 \\
 1 & -1 &  1 &  1 &  1 &  1 & -1 & -1 &  1 & -1 & -1 &  1 & -1 & -1 & -1 &  1 \\
 1 &  1 & -1 &  1 &  1 & -1 &  1 & -1 & -1 &  1 & -1 & -1 &  1 & -1 & -1 &  1 \\
 1 &  1 &  1 & -1 &  1 & -1 & -1 & -1 &  1 & -1 &  1 & -1 & -1 &  1 & -1 &  1 \\\hline
 1 &  1 &  1 &  1 & -1 &  1 &  1 &  1 & -1 & -1 & -1 & -1 & -1 & -1 &  1 &  1 \\
\end{array}
\right]
\]

\[
\setlength{\arraycolsep}{2pt}
A^{-1} = \frac12
\left[
\mbox{\scriptsize$
\begin{array}{r|rrrr|rr!{\vrule width 2pt}rrrr|rrrr|r}
-2&-95&-117&-195&-108&-50&-30&50&110&136&142&108&36&12&4&1\\\hline
1&91&114&192&108&52&32&-48&-108&-134&-140&-107&-36&-12&-4&-1\\
1&90&112&189&106&51&31&-47&-106&-132&-138&-105&-35&-12&-4&-1\\
1&85&106&178&100&48&29&-45&-100&-124&-130&-99&-33&-11&-4&-1\\
1&69&86&145&81&39&24&-36&-81&-101&-106&-81&-27&-9&-3&-1\\\hline
0&-67&-84&-143&-81&-40&-25&35&80&100&105&80&27&9&3&1\\
0&-65&-82&-139&-79&-39&-24&34&78&97&102&78&26&9&3&1\\
0&-60&-76&-129&-73&-36&-22&32&72&90&94&72&24&8&3&1\\
0&-82&-103&-175&-99&-49&-30&43&98&122&128&98&33&11&4&1\\
0&-87&-110&-186&-105&-52&-32&46&104&130&136&104&35&12&4&1\\
0&-80&-101&-172&-97&-48&-30&42&96&120&126&96&32&11&4&1\\\hline
0&57&72&123&70&35&22&-30&-69&-86&-90&-69&-23&-8&-3&-1\\
0&59&75&127&72&36&22&-31&-71&-89&-93&-71&-24&-8&-3&-1\\
0&64&81&138&78&39&24&-34&-77&-96&-101&-77&-26&-9&-3&-1\\
0&78&99&169&96&48&30&-41&-94&-118&-124&-95&-32&-11&-4&-1\\\hline
0&-57&-72&-122&-69&-34&-21&30&68&85&89&68&23&8&3&1
\end{array}$}
\right]
\]

\subsection{First column}
\label{SEC:FirstColumn}
We first solve the equation $Lz = e_1$. By induction we show that $z_j
= 1 - k$, whenever $|\alpha_j|=k$. For the base case, clearly $z_1=1$,
since the first row of $L$ is $e_1^\transpose$.  Next for the
induction step we treat the cases of $k=1$ and $k \geq 2$ separately.

Let $|\alpha_j|=1$. When $j=i_1=2$, we have $A_2=\{1\}$, and hence
$F_2=\{\alpha_1,\alpha_2\}$. When $j>i_1=2$, we have
$A_2=\{j-2,j-1\}$, and hence $F_j=\{\alpha_{j-1},\alpha_j\}$. In both
cases we actually have $F_j \setminus \{\alpha_{j-1},\alpha_j\} =
\emptyset$. Thus
\[
z_j = (2^{1-1}-1)z_{j-1} - \sum_{\alpha_\ell \in F_j \setminus
  \{\alpha_{j-1},\alpha_j\}} z_\ell = 0 \enspace .
\]

Let $|\alpha_j| = k \geq 2$. Consider first $j=i_k$. We have
$|A_j|=|\alpha_j|=k$, and $F_j$ contains $\binom{k}{s}$
sets of size $s$. By the induction hypothesis we have
\[
\begin{split}
z_j & = (2^{k-1}-1)z_{j-1} - \sum_{\alpha_\ell \in F_j \setminus  \{\alpha_{j-1},\alpha_j\}} z_\ell \\
& = (2^{k-1}-1)(1-(k-1)) - \sum_{s=0}^{k-2} \binom{k}{s} (1-s) - (k-1)(1-(k-1)) \\
& = (2^{k-1}-1)(2-k) - \sum_{s=0}^{k} \binom{k}{s} (1-s)  + (2-k) + (1-k) \\
& = (2^{k-1}-1)(2-k) - (2-k)2^{k-1} + (2-k) + (1-k) = 1-k
\end{split}
\]
Consider next $j>i_k$. We have $|A_j|=k+1$, and $F_j$ contains
$2\binom{k-1}{s-1}$ sets of size $s>0$. In particular $F_j$
contains only the sets $\alpha_j$ and $\alpha_{j-1}$ of size
$k$. Again by the induction hypothesis we have
\[
\begin{split}
z_j & = (2^{k-1}-1)z_{j-1} - \sum_{\alpha_\ell \in F_j \setminus  \{\alpha_{j-1},\alpha_j\}} z_\ell \\
& = (2^{k-1}-1)(1-k) - \sum_{s=0}^{k-2} 2\binom{k-1}{s}(1-(s+1)) \\
& = (2^{k-1}-1)(1-k) + 2 \sum_{s=0}^{k-1} \binom{k-1}{s}s - 2(k-1)\\
& = (2^{k-1}-1)(1-k) + 2 (k-1)2^{k-2} - 2(k-1) = 1-k
\end{split}
\]

\subsubsection{Sign pattern}
\label{SEC:FirstColumn-sign-pattern}

Next we compute $x = \frac{1}{n} Q z$. We need the following
well-known identity (\cite[Equation 5.42]{GrahamKnuthPatashnik})
\begin{lemma}
\label{LEM:BinomialInversionPolynomial}
Let $P(a)=c_0+c_1a+\dots+c_ka^k$ be a polynomial, $k\geq 0$. Then
\[
\sum_{a=0}^k \binom{k}{a} (-1)^a P(a) = (-1)^k k!\,c_k \enspace .
\]
\qed
\end{lemma}

Consider given $j$, and let $k=\abs{\alpha_j}$. Then
\begin{equation}
\label{EQ:x_j-sign-column1}
x_j =
\frac{1}{n} \sum_{\ell=0}^{n} (-1)^{\abs{\alpha_j \cap \alpha_\ell}} z_\ell = \frac{1}{n} \sum_{\ell=0}^{n} (-1)^{\abs{\alpha_j \cap \alpha_\ell}} (1-\abs{\alpha_\ell})
\end{equation}
For given $\ell$, let $a=\abs{\alpha_j \cap \alpha_\ell}$ and
$b=\abs{\alpha_\ell}-a$. Then we may collect the terms in
Equation~(\ref{EQ:x_j-sign-column1}) according to $a$ and $b$ and obtain
\[
x_j = 2^{-m}\sum_{a=0}^k \binom{k}{a} (-1)^a \sum_{b=0}^{m-k} \binom{m-k}{b} (1-a-b)
\]
We first evaluate the innermost summation
\[
\begin{split}
& \sum_{b=0}^{m-k} \binom{m-k}{b} (1-a-b)  = (1-a)\sum_{b=0}^{m-k} \binom{m-k}{b} - \sum_{b=0}^{m-k} \binom{m-k}{b}b \\
 & = (1-a)2^{m-k} - (m-k)2^{m-k-1} = 2^{m-k} ( 1-(m-k)/2-a)
\end{split}
\]
Thus,
\[
x_j=2^{-k}\sum_{a=0}^k \binom{k}{a} (-1)^a ( 1-(m-k)/2-a)
\]
By Lemma \ref{LEM:BinomialInversionPolynomial} we have $x_1=1-m/2$,
since $\abs{\alpha_1}=0$, $x_j=1/2$, for $k=\abs{\alpha_j}=1$, and
$x_j=0$ when $k=\abs{\alpha_j}>1$.

\subsection{Second block}
\label{SEC:SecondBlock}

Here we consider the equation $Lz = e_i$, for $\abs{\alpha_i}=1$. The
last column of the block is handled separately. In this and the
following sections we shall use the notation $z_\alpha$ as a shorthand
for the entry $z_{\num(\alpha)}$.

\subsubsection{First \texorpdfstring{$m-1$}{m-1} columns}

We have $z_1=0$, since the first row of $L$ is $e_1^\transpose$. Next,
$z_j = 0$ for $j \in \{2,\dots,m+1\} \setminus \{i\}$, and $z_i=1$,
since row $j$ of $L$ is $e_j^\transpose$ for all $j \in \{2,\dots,m+1\}$.

For $j=i_2$ we have the equation
\[
\frac{1}{2}(z_\emptyset+z_{\{1\}}+(1-2)z_{\{m\}}+z_{i_2})=0 \enspace ,
\]
since $A_{i_2}=\{1,m\}$, and
$F_{i_2}=\{\emptyset,\{1\},\{m\},\{1,m\}\}$. We assume here that
$i<m+1=\num(\{m\})$, and thus have $z_{i_2} = -z_{\{1\}}$. Thus
$z_{i_2}=0$ when $i>2$, and $z_{i_2}=-1$ when $i=2$. For $i_2<j<i_3$
we have $\abs{\alpha_j \bigtriangleup \alpha_{j-1}}=2$. Let $\alpha_j
\bigtriangleup \alpha_{j-1}=\{a,b\}$. Then
$F_j=\{\{a\},\{b\},\alpha_{j-1},\alpha_j\}$ and we have the equation
\[
\frac{1}{2}(z_{\{a\}} + z_{\{b\}} + (1-2) z_{j-1} + z_j)=0 \enspace .
\]
Hence $z_j=z_{j-1}-1$ if $i \in \alpha_j \bigtriangleup \alpha_{j-1}$,
and $z_j=z_{j-1}$ otherwise.

Note that any element of $\{2,\dots,m-2\}$ appears 2 times in the
symmetric differences $\alpha_{j-1} \bigtriangleup \alpha_{j}$ for
$i_2<j\leq \num(\{m-1,m\})$, whereas the elements $1$ and $m-1$ appear
only 1 time. Hence it follows that $z_{\{m,m-2\}}$ and $z_{\{m,m-1\}}$
are both at least $-2$.

Note that when $m \geq 6$, any element of $\{2,\dots,m-1\}$ appears at
least 4 times in the symmetric differences $\alpha_{j-1}
\bigtriangleup \alpha_{j}$ for $i_2<j<i_3$, whereas the element $1$
appear exactly $3$ times. In both cases this implies by the above that
$z_{i_3-1} \leq -4$.

Consider now $j=i_3$. Then $A_j = \{m-2,m-1,m\}$ and $F_j =
\big\{\emptyset,\{m-2\},\{m-1\},\{m\},\{m-2,m-1\},\{m-2,m\},\{m-1,m\},\{m-2,m-1,m\}\big\}$,
and we thus have the equation
\[
\frac{1}{4}(z_\emptyset+z_{\{m-2\}}+z_{\{m-1\}}+z_{\{m\}}+z_{\{m-2,m\}}+z_{\{m-1,m\}}+(1-4)z_{i_3-1}+z_{i_3})=0 \enspace ,
\]
and hence $z_{i_3} = 3 z_{i_{3}-1} - z_{\{m,m-2\}} - z_{\{m,m-1\}}$,
when $i \notin \{m-2,m-1\}$, and $z_{i_3} = 3 z_{i_{3}-1} -
z_{\{m,m-2\}} - z_{\{m,m-1\}} - 1$, otherwise. In both cases $z_{i_3}
\leq 3 z_{i_{3}-1} - z_{\{m,m-2\}} - z_{\{m,m-1\}}$. We already have
the estimates $z_{i_3-1} \leq -4$, and $z_{\{m,m-2\}}, z_{\{m,m-1\}}
\geq -2$. Hence $z_{\{m,m-2\}}+z_{\{m,m-1\}} \geq -4\geq z_{i_3-1}$,
and it follows
\[
z_{i_3} \leq 2 z_{i_3-1} \leq -8 \enspace .
\]

Next consider $j=i_3+1$. Then $A_j = \{m-3,m-2,m-1,m\}$, but the set
$F_j$ depends on whether $m$ is even or odd. Write $A_j =
\{m-3,m-1,a,b\}$, where $\alpha_{i_3} \symdiff \alpha_{i_3+1} =
\{m-3,a\}$.  Then $F_j =
\big\{\{m-3\},\{a\},\{m-3,m-1\},\{m-3,b\},\{a,m-1\},\{m-2,m\},\{m-3,m-1,b\},\{m-2,m-1,m\}\big\}$. Let $\beta_1,\dots,\beta_4$ denote the sets of size 2 of $F_j$. We then have the equation
\[
\frac{1}{4}(z_{\{m-3\}} + z_{\{a\}} + z_{\beta_1} + z_{\beta_2} + z_{\beta_3} + z_{\beta_4} + (1-4)z_{i_3} + z_{i_3+1}) = 0
\]
and hence $z_{i_3+1} = 3 z_{i_{3}} - \sum_{s=1}^4 z_{\beta_s}$,
when $i \notin \{m-3,a\}$, and $z_{i_3+1} = 3 z_{i_{3}} -
\sum_{s=1}^4 z_{\beta_s} - 1$, otherwise. In both cases
$z_{i_3+1} \leq 3 z_{i_{3}} - \sum_{s=1}^4 z_{\beta_s}$. As seen
$z_{\{m-2,m\}} \geq -2 \geq \frac{1}{2}z_{i_3-1}$, and for the other
$\beta_s$ we have $z_{\beta_s} \geq z_{i_3-1}$. Thus
\[
z_{i_3+1}
\leq 3z_{i_3} - (3+\frac{1}{2})z_{i_3-1} \leq 3z_{i_3}
-(3+\frac{1}{2})\frac{1}{2}z_{i_3} = \frac{5}{4} z_{i_3} \enspace .
\]
For the remaining $j=i_3+\ell+1$, $\ell > 0$, for which
$\abs{\alpha_j}=3$ we have the similar inequality
\[
z_{i_3+\ell+1} \leq 3 z_{i_3+\ell} - \sum_{s=1}^4 z_{\beta_s} \leq 3
z_{i_3+\ell} - 4 z_{i_3-1} \leq 3 z_{i_3+\ell} - 2 z_{i_3} \leq 3
z_{i_3+\ell} - \frac{8}{5} z_{i_3+1} \enspace ,
\]
for appropriate sets $\beta_1,\dots,\beta_4$ of size 2.

\begin{claim}
\label{CLA:zifactor}
For $j=i_3+\ell+1$, $\ell > 0$, for which $\abs{\alpha_j}=3$ we have
\[
z_{i_3+\ell+1} \leq \frac{3^\ell + 4}{3^{\ell-1}+4}z_{i_3+\ell}
\]
\end{claim}
\begin{proof}
  The proof is by induction on $\ell$. For $\ell=1$ we have
\[
z_{i_3+2} \leq 3z_{i_3+1} - \frac{8}{5}z_{i_3+1}z_{i_3+1} = \frac{3+4}{1+4}z_{i_3+1} \enspace .
\]
Next, for the induction step
\[
\begin{split}
z_{i_3+\ell+1} & \leq 3 z_{i_3+\ell} - \frac{8}{5}z_{i_3+1} \leq 3 z_{i_3+\ell} - \frac{8}{5}\left(\prod_{s=1}^{\ell-1} \frac{3^{s-1}+4}{3^s+4}\right)z_{i_3+\ell} \\ & = (3-\frac{8}{5}\frac{5}{3^{\ell-1}+4})z_{i_3+\ell} = \frac{3^\ell +4}{3^{\ell-1}+4}z_{i_3+\ell} \enspace .
\end{split}
\]
\end{proof}
We next find $s$ such that
\[
\abs{z_{i_3+s}} > \sum_{\ell=1}^{i_3+s-1} \abs{z_{\ell}} \enspace .
\]
We estimate the first $i_3+3$ terms separately. We have at most
$1+m(m-1)/2+4$ nonzero terms, each of absolute value less than
$\abs{z_{i_3+4}}$. That is
\[
\sum_{\ell=1}^{i_3+3} \abs{z_\ell} < m^2\abs{z_{i_3+4}} \enspace ,
\]
using $m\geq 6$.

Using Claim~\ref{CLA:zifactor}, observing that $\frac{3^\ell +
  4}{3^{\ell-1}+4}$ is increasing with $\ell$, and $\frac{3^4 +
  4}{3^{4-1}+4} \geq \frac{5}{2}$, we can apply
Lemma~\ref{LEM:fastgrowth} with $w_i=\abs{z_{i_3+3+i}}$, $c=m^2$,
$\epsilon=\frac{1}{2}$ to obtain
\[
\abs{z_{i_3+s}} > \sum_{\ell=1}^{i_3+s-1} \abs{z_{\ell}}
\]
for $s=\log_{5/2}(2m^2)+5$. Note that $i_3+s < i_4$, since there are
$\binom{m}{3}$ sets of size $3$, and $m\geq 6$. Also by
Claim~\ref{CLA:zifactor} we have that $z_{i_3+s}<0$. By
Lemma~\ref{LEM:fastgrowth2} we thus have
\[
z_n < -\sum_{\ell=1}^{n-1} \abs{z_\ell} \enspace .
\]

\subsubsection{Last column}

Here we consider the last column of the second block, corresponding to
solving the equation $Lz = e_{i_2-1}$. As above we have $z_1=0$, and
$z_j = 0$ for $j \in \{2,\dots,m\}$, and $z_{i_2-1}=1$. For $j=i_2$ we
have the equation
\[
\frac{1}{2}(z_\emptyset+z_{\{1\}}+(1-2)z_{\{m\}}+z_{i_2})=0 \enspace ,
\]
since $A_{i_2}=\{1,m\}$, and
$F_{i_2}=\{\emptyset,\{1\},\{m\},\{1,m\}\}$. It follows that
$z_{i_2}=1$.  Also as above for $i_2<j<i_3$, $z_j=z_{j-1}$ if $m \in
\alpha_j \symdiff \alpha_{j-1}$, and $z_j=z_{j-1}$ otherwise. All sets
of size $2$ containing the element $m$ comes before all other sets of
size $2$ in the order, and hence the case of $m \in \alpha_j \symdiff
\alpha_{j-1}$ occurs only when $j=\num(\{m,m-1\})+1=2m+1$. Thus we have
$z_j = 1$ when $i_2<j\leq 2m$ and $z_j = 0$ when $2m<j<i_3$.

For $j=i_3$, we have the equation
\[
\frac{1}{4}(z_\emptyset+z_{\{m-2\}}+z_{\{m-1\}}+z_{\{m\}}+z_{\{m-2,m\}}+z_{\{m-1,m\}}+(1-4)z_{i_3-1}+z_{i_3})=0 \enspace ,
\]
since again $A_j = \{m-2,m-1,m\}$ and $F_j =
\big\{\emptyset,\{m-2\},\{m-1\},\{m\},\{m-2,m-1\},\{m-2,m\},\{m-1,m\},\{m-2,m-1,m\}\big\}$. From
this we see that $z_{i_3}=-3$.

For $i_3<j<i_4$ we have
\[
z_j = (2^{3-1}-1) z_{j-1} - \sum_{\alpha_\ell \in F_j\setminus
  \{\alpha_{j-1},\alpha_j\}} z_\ell \leq 3 z_{j-1} \enspace ,
\]
and also $z_j<0$. Again we find $s$ such that
\[
\abs{z_{i_3+s}} > \sum_{\ell=1}^{i_3+s-1} \abs{z_{\ell}} \enspace .
\]
Note that $\sum_{\ell=1}^{i_3-1} \abs{z_\ell} = 1+(m-1)=m$. We can
then apply Lemma~\ref{LEM:fastgrowth} with $w_i=\abs{z_{i_3+i-1}}$,
$c=m$, $\epsilon=1$ to obtain
\[
\abs{z_{i_3+s}} > \sum_{\ell=1}^{i_3+s-1} \abs{z_{\ell}}
\]
for $s=\log_{3}(m)+2$, noting that $i_3+s<i_4$. We also have
$z_{i_3+s}<0$. By Lemma~\ref{LEM:fastgrowth2} we then have
\[
z_n < -\sum_{\ell=1}^{n-1} \abs{z_\ell} \enspace .
\]

\subsection{Third block}
\label{SEC:ThirdBlock}

Here we consider the equation $Lz = e_i$, for $\abs{\alpha_i}=2$. We
have $z_j=0$ for $j<i$, $z_i=1$, and hence $z_j=(2^{2-1}-1)z_{j-1} -
\sum_{\alpha_\ell \in F_j\setminus \{\alpha_{j-1},\alpha_j\}}
z_\ell=1$, for $i<j<i_3$. We remark for future use that
$\sum_{\ell=1}^{i_3-1} \abs{z_\ell} \leq \binom{m}{2}$.

Consider now $j=i_3$. As usual we have the equation
\[
\frac{1}{4}(z_\emptyset+z_{\{m-2\}}+z_{\{m-1\}}+z_{\{m\}}+z_{\{m-2,m\}}+z_{\{m-1,m\}}+(1-4)z_{i_3-1}+z_{i_3})=0 \enspace ,
\]
and hence
\[
z_{i_3} = 3 z_{i_3-1} - z_{\{m-2,m\}} - z_{\{m-1,m\}} = 3 - z_{\{m-2,m\}} - z_{\{m-1,m\}}.
\]

Next consider $j=i_3+1$. As above $A_j = \{m-3,m-2,m-1,m\}$, where the
set $F_j$ depends on whether $m$ is even or odd. Write again $A_j =
\{m-3,m-1,a,b\}$, where $\alpha_{i_3} \symdiff \alpha_{i_3+1} =
\{m-3,a\}$.  Then $F_j =
\big\{\{m-3\},\{a\},\{m-3,m-1\},\{m-3,b\},\{a,m-1\},\{m-2,m\},\{m-3,m-1,b\},\{m-2,m-1,m\}\big\}$. Let
$\beta_1,\dots,\beta_4$ denote the sets of size 2 of $F_j$. We then
have the equation
\[
\frac{1}{4}(z_{\{m-3\}} + z_{\{a\}} + z_{\beta_1} + z_{\beta_2} + z_{\beta_3} + z_{\beta_4} + (1-4)z_{i_3} + z_{i_3+1}) = 0 \enspace ,
\]
and hence $z_{i_3+1} = 3z_{i_3} - \sum_{s=1}^4z_{\beta_s}$. We consider
below 3 cases depending on the relationship between $i$ and
$\num(\{m-2,m\})=2m-1$.

\begin{itemize}
\item $i<\num(\{m-2,m\})$: Here $z_{i_3} = 1$. Also in our ordering,
  for $s\in \{1,\dots,4\}$ we have $\num(\beta_s) \geq \num(\{m-3,m\})
  \geq i$. It follows that $z_{i_3+1}=-1$. For $i_3+1<j<i_4$, we have
  $z_j=(2^{3-1}-1)z_{j-1} - \sum_{\alpha_\ell \in F_j\setminus
    \{\alpha_{j-1},\alpha_j\}} z_\ell \leq 3z_{j-1}$. We can thus
  apply Lemma~\ref{LEM:fastgrowth} with $w_i=\abs{z_{i_3+i}}$,
  $c=m^2/2$, $\epsilon=1$ to obtain
\[
\abs{z_{i_3+s}} > \sum_{\ell=1}^{i_3+s-1} \abs{z_{\ell}}
\]
for $s=\log_{3}(m^2/2)+2$, noting that $i_3+s < i_4$. We also have
$z_{i_3+s}<0$. By Lemma~\ref{LEM:fastgrowth2} we then have
\[
z_n < -\sum_{\ell=1}^{n-1} \abs{z_\ell} \enspace .
\]
\item $i=\num(\{m-2,m\})$: Here $z_{i_3} = 1$ as well. In case $a=m$
  we have $z_{i_3+1}=-1$. Also
\[
z_{i_3+2} \leq (2^{3-1}-1)z_{i_3+1} -  \sum_{\alpha_\ell \in F_{i_3+2}\setminus
    \{\alpha_{i_3+1},\alpha_{i_3+2}\}} z_\ell \leq 3z_{i_3+1} =
  -3 \enspace .
\]
However in case $a=m-2$ we have $z_{i_3+1}=0$. Thus we consider
$j=i_3+2$ as well. We have here that $\alpha_j = \{m-3,m-2,m\}$ and
thus $\alpha_{i_3+1}\symdiff \alpha_{i_3+2} = \{m-1,m\}$. Then
\[
z_{i_3+2} = -z_{\{m-3,m-1\}}-z_{\{m-2,m-1\}}-z_{\{m-3,m\}}-z_{\{m-2,m\}} = -3
\]
For $i_3+2<j<i_4$, we have as before $z_j \leq 3z_{j-1}$, and can thus
apply Lemma~\ref{LEM:fastgrowth} with $w_i=\abs{z_{i_3+i+1}}$,
$c=m^2/2$, $\epsilon=1$ to obtain
\[
\abs{z_{i_3+s}} > \sum_{\ell=1}^{i_3+s-1} \abs{z_{\ell}}
\]
for $s=\log_{3}(m^2/2)+3$, noting that $i_3+s< i_4$. We also have
$z_{i_3+s}<0$. By Lemma~\ref{LEM:fastgrowth2} we then have
\[
z_n < -\sum_{\ell=1}^{n-1} \abs{z_\ell} \enspace .
\]

\item $i>\num(\{m-2,m\})$: Here $z_{i_3}\geq 2$.  Since
  $\num(\{m-2,m\})<\num(\{m-1,m\}) \leq i$ we have $z_{i_3+1} = 3
  z_{i_{3}} - \sum_{s=1}^4 z_{\beta_s} \geq 3z_{i_3} - 3 \geq
  \frac{3}{2}z_{i_3}$. Let $\beta^{(2)}_1,\dots,\beta^{(2)}_4$ denote
  the sets of size 2 of $F_{i_3+2}$. Then $z_{i_3+2} = 3 z_{i_3+1} -
  \sum_{s=1}^4 z_{\beta^{(2)}_s} \geq \frac{5}{3}z_{i_3+1} + 2z_{i_3}
  -4 \geq \frac{5}{3}z_{i_3+1}$. Again, let
  $\beta^{(3)}_1,\dots,\beta^{(3)}_4$ denote the sets of size 2 of
  $F_{i_3+3}$. Then $z_{i_3+3} = 3 z_{i_3+2} - \sum_{s=1}^4
  z_{\beta^{(3)}_s} \geq \frac{11}{5}z_{i_3+2} +
  \frac{4}{5}\cdot\frac{5}{3}\cdot\frac{3}{2}z_{i_3}-4 \geq
  \frac{11}{5}z_{i_3+2}$. By induction it is now easy to derive $z_j
  \geq \frac{11}{5} z_{j-1}$ for $i_3+3<j<i_4$. We can thus apply
  Lemma~\ref{LEM:fastgrowth} with $w_i=z_{i_3+i+1}$, $c=m^2/2$,
  $\epsilon=\frac{1}{5}$ to obtain
\[
\abs{z_{i_3+s}} > \sum_{\ell=1}^{i_3+s-1} \abs{z_{\ell}}
\]
for $s=\log_{\frac{11}{5}}(m^2/2)+3$, noting that $i_3+s< i_4$. We also have
$z_{i_3+s}>0$. By Lemma~\ref{LEM:fastgrowth2} we then have
\[
z_n > \sum_{\ell=1}^{n-1} \abs{z_\ell} \enspace .
\]
\end{itemize}

\subsection{Remaining columns}
\label{SEC:RemainingColumns}

Here we consider the equation $Lz = e_i$, for $\abs{\alpha_i}=k\geq
3$.  We then have that $z_j=0$ for $j<i$, $z_i=2^{k-1} \geq
4$. Clearly $z_i > \sum_{\ell=1}^{i-1} \abs{z_\ell} = 0$. By
Lemma~\ref{LEM:fastgrowth2} we then have
\[
z_n > \sum_{\ell=1}^{n-1} \abs{z_\ell} \enspace .
\]


\subsection{Explicitness}
\label{SEC:AlonVuExplicitness}

We discuss here in more detail the explicitness of our
construction. We say that a family $\{A_n\}$ of matrices, where $A_n$
is a $n \times n$ matrix, is \emph{explicit}, if there is an algorithm
that given as input $n$ computes the matrix $A_n$ in time polynomial
in $n$. We say that the family is \emph{fully-explicit}, if there is
an algorithm that given as input $i$, $j$, and $n$, computes entry
$(i,j)$ of $A_n$ in time polynomial in $\log(n)$. Clearly the latter
definition is more restrictive than the former.

We next give a sketch of an argument that the matrices just
constructed are fully-explicit. Let $A$ be the $2^m \times 2^m$ matrix
of Section \ref{SEC:Alon-Vu} with the order of Section
\ref{SEC:Ordering}, and let $B$ be the $2^m \times 2^m$ matrix
obtained by the Hadamard product of $A$ with the transpose of the sign
pattern $\Sigma$ of Section \ref{SEC:SignPattern}. In order to compute
entry $(i,j)$ of $B$ we compute separately the entry $(i,j)$ of $A$
and the entry $(j,i)$ of $\Sigma$.

To compute entry $(j,i)$ of $\Sigma$ we need to determine which block
of rows that row $j$ belongs to and to check whether $i>2m-1$. The
former is determined by finding $k$ such that
\[
\sum_{\ell=0}^k \binom{m}{\ell} < j \leq \sum_{\ell=0}^{k+1}
\binom{m}{\ell} \enspace .
\]
This is easily done in time polynomial in $m$.

To compute entry $(i,j)$ of $A$ it is sufficient to observe that the
following task can be computed in time polynomial in $m$: Given index
$i$, compute the set of the order, $\alpha_i$. To do this, first
compute the $k$ such that $\abs{\alpha_i}=k$. This is the same task as
just considered, and identifies the order $\beta^k_m$. Depending on
$k$ we may need to consider the reverse of this. By appropriately
adjusting $i$, we may just consider consider finding the set $i'$ of
the order $\beta^k_m$ (for $k=2$ we also need to \emph{shift} the set
afterwards). This can be done by first identifying the smallest
element $a$ of the set $\alpha_{i'}$ and recursing. Specifically, $a$
can be determined by the inequalities
\[
\sum_{\ell=1}^{a-1} \binom{m-\ell}{k-1} < i' \leq \sum_{\ell=1}^{a} \binom{m-\ell}{k-1} \enspace .
\]
Then $i'$ is adjusted by subtracting the sum $\sum_{\ell=1}^{a-1}
\binom{m-\ell}{k-1}$, and continuing with $\beta^{k-1}_{m-a}$,
possibly adjusting $i'$ again if the order is reversed, finding the
next-smallest element and so on.

\bibliographystyle{abbrv}
\bibliography{MatrixGamePatience}

\end{document}